\documentclass{amsart}
\usepackage{amsmath,amsfonts,amssymb,amsthm,mathtools}
\usepackage{amscd}
\usepackage{bbm}
\usepackage{enumerate,enumitem}
\usepackage{galois}
\usepackage{mathrsfs}
\usepackage{xypic}
\usepackage{geometry}
\usepackage{hyperref}
\usepackage{tikz}

\usepackage{color}

\theoremstyle{plain}
\newtheorem{Th}{Theorem}[section]
\newtheorem{Cor}[Th]{Corollary}
\newtheorem{Lem}[Th]{Lemma}
\newtheorem{Prop}[Th]{Proposition}

\newtheorem{Ex}{Example}[section]
\newtheorem{Rem}{Remark}[section]

\numberwithin{equation}{section}

\newcommand{\diff}[2]{\frac{\partial #1}{\partial #2}}

\newcommand{\qa}{\alpha}
\newcommand{\qb}{\beta}
\newcommand{\qd}{\delta}
\newcommand{\qg}{\gamma}
\newcommand{\qs}{\sigma}

\newcommand{\qp}{\partial}

\newcommand{\ql}{\lambda}

\newcommand{\Qd}{\Delta}

\newcommand{\vard}[2]{\frac{\delta #1}{\delta #2}}

\newcommand{\kk}[1]{\left(#1\right)}

\newcommand{\fk}[2]{\left[#1, #2\right]}
\newcommand{\vac}{|0\rangle}
\newcommand{\hes}[1]{b_{(#1)}}
\begin{document}

\title{Quantum dispersionless KdV hierarchy revisited}

\author{Zhe Wang}
\address{RIKEN Center for Interdisciplinary Theoretical and Mathematical Sciences (iTHEMS), RIKEN, Wako 351-0198, Japan}
\email{zhe.wang.aa@riken.jp}

\begin{abstract}
We quantize Hamiltonian structures with hydrodynamic leading terms using the Heisenberg vertex algebra. As an application, we construct the quantum dispersionless KdV hierarchy via a non-associative Weyl quantization and compute the corresponding eigenvalue problem. 
\end{abstract}

\date{\today}

\maketitle
\tableofcontents

%%%%%%%%%%%%%%%%%%%%%%%%%%%%%%%%%%%%%%%%%%%%%%%
%%%%%%%%%%%%%%%%%%%%%%%%%%%%%%%%%%%%%%%%%%%%%%%
%%%%%%%%%%%%%%%%%%%%%%%%%%%%%%%%%%%%%%%%%%%%%%%

\section{Introduction}

Since the proof of the Witten-Kontsevich theorem \cite{kontsevich1992intersection,witten1990two}, integrable hierarchies have provided important tools for studying various problems in mathematical physics. For example, in 2d topological field theories, the partition functions are controlled by integrable hierarchies of evolutionary Hamiltonian PDEs \cite{buryak2015double,dubrovin2001normal}. In symplectic field theories \cite{eliashberg2010introduction}, the partition functions satisfy a family of Schr\"odinger type equations. The quantum Hamiltonians involved in these Schr\"odinger equations
 mutually commute, and therefore can be viewed as quantum integrable hierarchies. While Hamiltonian integrable hierarchies as well as their relationships to 2d topological field theories have been studied in great detail during the last three decades, the corresponding quantum counterparts are still under investigation.

An important advancement toward a general theory of quantum integrable hierarchies was achieved by Buryak and Rossi \cite{buryak2016double}, where a quantization framework is established for arbitrary cohomological field theories using the geometry of double ramification cycles. In their constructions, quantum Hamiltonian structures are given by the canonical quantization of the classical Hamiltonian structures of hydrodynamic types, and quantum Hamiltonians are defined through the intersection theory which is specified by the given cohomological field theory. In particular, their work revealed very deep relationships between quantum integrable hierarchies and geometries of the moduli space of stable curves.

This paper approaches quantum integrable hierarchies in terms of vertex algebras by studying the quantization of the dispersionless KdV hierarchy. The quantization of such a particular integrable hierarchy was addressed both in the framework of symplectic field theory \cite{eliashberg2010introduction,rossi2008gromov} and of the double ramification cycle theory \cite{buryak2016double}, and its properties were studied in a series of works \cite{dubrovin2016symplectic,van2024quantum,ittersum2024quantum,ruzza2021spectral}. In the present paper, however, a quantization framework different from above-mentioned ones is established. In particular, we realize a quantum Hamiltonian structure of hydrodynamic type as the Lie algebra canonically attached to the Heisenberg vertex algebra, and apply it to study the quantization of the dispersionless KdV hierarchy. As a result, we construct an Abelian Lie subalgebra consisting of certain states in the Heisenberg vertex algebra, and classical limits of these states coincide with classical Hamiltonian densities of the dispersionless KdV hierarchy. These states provide a quantization of the classical hierarchy, and their properties are studied in detail. 

It should be emphasized that the quantum dispersionless KdV hierarchy constructed in this paper looks different from the one studied in the above-mentioned papers. Both the quantum Hamiltonian structure and the quantum Hamiltonians are different. However,   
despite these differences, the quantization framework based on vertex algebras shares great similarities with the framework based on geometries of moduli spaces. A comparison between these two constructions is given in Sect.\,\ref{AJ}.

Let us illustrate the basic concepts regarding quantization of Hamiltonian integrable hierarchies and present main results of this paper. We begin by introducing necessary notations. Define the ring of differential polynomial
\[
\mathcal A = \mathbb C\left[v,v^{(1)},v^{(2)},\dots\right],
\]
and denote by 
\[
\qp_x = \sum_{n\geq 0}v^{(n+1)}\diff{}{v^{(n)}}, \quad v^{(0)}:=v,
\] 
a derivation on $\mathcal A$. It follows that we can formally view $v$ as a function $v(x)$ with the independent variable $x$, and we have $v^{(n)} = \qp_x^n v(x)$. Furthermore, we introduce the quotient space $\mathcal F = \mathcal A/\qp_x\mathcal A$ whose elements are called local functionals. For an element $f\in\mathcal A$, denote by $F = \int f$ its class in $\mathcal F$, and we call $f$ a density of $F$. Using these notions, the (classical) dispersionless KdV hierarchy can be presented by 
\[
\diff{v}{t_n} = \frac{v^n}{n!}\diff{v}{x}=\diff{}{x}\left(\vard{H_{n+1}^c}{v}\right),\quad H_n^c=\int \frac{v^{n+1}}{(n+1)!},\quad n\geq 0,
\]
where the variational derivative is defined by 
\[
\vard{F}{v}:=\sum_{n\geq 0}(-\qp_x)^n\diff{f}{v^{(n)}},\quad F = \int f\in \mathcal F.
\]
Note that the variational derivative of a local functional is independent of the choice of its density. 

One of the most important properties of the dispersionless KdV hierarchy is that it is Hamiltonian. We define the classical Hamiltonian structure of the dispersionless KdV hierarchy to be the Lie algebra structure on $\mathcal F$ specified by 
\begin{equation}
    \label{AC}
\left\{F,G\right\} = \int \vard{F}{v}\qp_x\vard{G}{v}.
\end{equation}
For a given local functional $F$, this Lie bracket induces a derivation on $\mathcal A$ given by 
\[
\{-,F\}\colon\mathcal A\to\mathcal A,\quad 
\{g,F\} = \sum_{n\geq 0}\diff{g}{v^{(n)}}\qp_x^{n+1}\vard{F}{v},\quad g\in\mathcal A.
\]
It follows from the definition that for any $F,G\in\mathcal F$ and any density $g\in\mathcal A$ of $G$, we have 
\[
\int \{g,F\} = \{G,F\},
\]
hence we abuse the notation and use the same bracket to denote both the Lie bracket on $\mathcal F$ and its induced action on $\mathcal A$. Using these notations, the dispersionless KdV hierarchy can be represented by the following Hamiltonian form:
\[
\diff{v}{t_n} = \{v,H_n^c\}.
\]
We call $H_n^c$ the classical Hamiltonians of the dispersionless KdV hierarchy. They are in involution, namely they span an Abelian Lie subalgebra of $\mathcal F$:
\[
\{H_n^c,H_m^c\} = 0,\quad n,m\geq 0.
\]

The Hamiltonian property of the dispersionless KdV hierarchy enables us to formulate its deformation quantization as follows:
%\begin{enumerate}[wide, labelwidth=!, labelindent=0pt]
\begin{enumerate}[leftmargin=*]
    \item Construct a Lie algebra bracket $[-,-]$ on $\mathcal F[\hbar]$ which deforms the classical Hamiltonian structure in the sense that
    \[
    \lim_{\hbar\to 0}\frac{1}{\hbar}[F,G] = \{F,G\},\quad F,G\in\mathcal F.
    \]
    This Lie algebra structure is called the quantum Hamiltonian structure.
    \item Deform the classical Hamiltonians $H_n^c$ into quantum Hamiltonians $H_n\in\mathcal F[\hbar]$ such that 
    \[
    [H_n,H_m] = 0,\quad \lim_{\hbar\to 0}H_n = H_n^c,\quad n,m\geq 0.
    \]
    We call the quantum Hamiltonians $H_n$ as well as the quantum Hamiltonian structure a quantization of the dispersionless KdV hierarchy.
    \item Solve the corresponding Schr\"odinger equations, i.e., determine a representation of the Lie algebra $\mathcal F[\hbar]$ and study the eigenvalue problem of $H_n$. 
\end{enumerate}
In this paper, we study all these three procedures in a unified framework based on the theory of vertex algebras.
Let us denote by $V$ the state space of the Heisenberg vertex algebra generated by the state $b\in V$, whose precise definition and related properties will be given in Sect.\,\ref{RAB} (or see Chapter 2 of \cite{frenkel2004vertex}). Associated with $b$, there is a formal power series of operators forming a basis of the Heisenberg Lie algebra:
\[
Y(b,z)= \sum_{n\in\mathbb Z} b_{(n)}z^{-n-1},\quad b_{(n)}\in \mathrm{End}(V),\quad \left[b_{(n)},b_{(m)}\right] = \hbar\,n\,\qd_{m+n,0},
\]
here the Lie bracket is the commutator on $\mathrm{End}(V)$.
Denote by $\vac\in V$ the vacuum state, then the state space $V$ is a free $\mathbb C[\hbar]$-module spanned by states of the form
\[
\hes{-k_1-1}\dots \hes{-k_n-1}\vac,\quad k_i\geq 0,\quad n\geq 0.
\]
For each state $a\in V$, it follows from the definition of a vertex algebra that we have a corresponding field $Y(a,z)$ given by 
\[
Y(a,z) = \sum_{n\in\mathbb Z}a_{(n)}z^{-n-1},\quad a_{(n)}\in \mathrm{End}(V),
\]
where $a_{(n)}$ is called the $n$-th mode of $a$. In particular, the $n$-th mode of the state $b$ is exactly $b_{(n)}$.
The general formula for $a_{(n)}$ in terms of modes of $b$ is given in Prop.\,\ref{AP}. We view the unknown function $v(x)$ of the dispersionless KdV hierarchy as a classical field, and we treat $Y(b,z)$ as the corresponding quantum field. It follows that there is a natural $\hbar$-linear isomorphism of vector spaces
\[
\varphi:\mathcal A[\hbar]\to V,\quad  v^{(k_1)}\dots  v^{(k_n)}\mapsto  k_1!\dots k_n!\,\hes{-k_1-1}\dots \hes{-k_n-1}\vac.
\]
Let $T\in\mathrm{End}(V)$ be the infinitesimal translation operator of the Heisenberg vertex algebra, which is characterized by 
\[
T\vac =0,\quad \left[T,b_{(n)}\right] = -nb_{(n-1)},\quad n\in\mathbb Z.
\]
Hence, it follows that 
\[\varphi\comp \qp_x = T\comp \varphi,\]
where the action of $\qp_x$ is extended to $\mathcal A[\hbar]$ by $\qp_x\hbar = 0$. As a result, the isomorphism $\varphi$ induces a natural isomorphism between the quotient spaces $\mathcal F[\hbar]$ and $V/TV$ which we still denote by $\varphi$. 

Using the above notations, we present here the main results of the paper. First, we construct a quantum Hamiltonian structure based on the vertex algebra structure. 
\begin{Th}
    \label{AF}
    The Lie bracket
    \begin{equation}
        \label{AE}
    [F,G] = \varphi^{-1}\left(\varphi(G)_{(0)}\varphi(F)\right),\quad F,G\in \mathcal F[\hbar],
    \end{equation}
    defines a deformation quantization of the Hamiltonian structure \eqref{AC}, here $\varphi(G)_{(0)}$ is the zeroth mode of the state $\varphi(G)$.  In terms of differential polynomials, this bracket is of the form
    \begin{align*}
    [F,G] = \sum_{m\geq 1}\frac{\hbar^m}{m!}\sum_{\substack{r_1,\dots,r_m\geq 0\\ s_1,\dots,s_m\geq 0}}&\frac{\qp^m f}{\qp v^{(s_1)}\dots \qp v^{(s_m)}}\frac{(-1)^{\sum r_i}\prod (r_i+s_i+1)!}{(2m-1+\sum r_i+\sum s_i)!}\\*
    &\cdot\qp_x^{2m-1+\sum r_i+\sum s_i}\frac{\qp^m g}{\qp v^{(r_1)}\dots\qp v^{(r_m)}},\quad F,G\in\mathcal F,
    \end{align*}
here $f, g\in \mathcal A$ are arbitrary densities of $F$ and $G$, respectively.
\end{Th}

Next, we construct the quantum Hamiltonians $H_n$ for the quantum  dispersionless KdV hierarchy. Let us denote by $h_n\in\mathcal A[\hbar]$ the differential polynomials determined by the recursion relation
\[
h_0 = \varphi^{-1}\kk{\hes{-1}\vac},\quad h_n = \frac{1}{n}\sum_{k=0}^{n-1}\frac{1}{\binom{n+1}{k+1}}\varphi^{-1}\left(\varphi(h_k)_{(-1)}\varphi(h_{n-k-1})\right),\quad n\geq 1.
\]
We call this a non-associative Weyl quantization. Very roughly speaking, $h_n$ is obtained by averaging the non-associativity of the Heisenberg vertex algebra. See Sect.\,\ref{AD} for details and a combinatorial way for constructing $h_n$. Though these quantum Hamiltonians are defined via the Heisenberg vertex algebra, it turns out that they can be computed effectively.
\begin{Prop}
    \label{AW}
    We have $h_0 = v$ and for $n\geq 1$
    \begin{align}
        \label{BU}
    h_n =\frac{1}{n(n+1)}&\left(\sum_{s,t\geq 0}\frac{(s+t+1)!}{s!t!}v^{(s)}v^{(t)}\diff{}{v^{(s+t)}}\right.\\
    \notag
    &\left.+\hbar \sum_{s,t\geq 0} \frac{(s+1)!(t+1)!}{(s+t+2)!}v^{(s+t+2)}\frac{\qp^2}{\qp v^{(s)}\qp v^{(t)}}\right)h_{n-1}. 
    \end{align}
\end{Prop}
\noindent
It follows that differential polynomials $h_n$ indeed define quantum Hamiltonians, and hence we construct a quantum dispersionless KdV hierarchy.
\begin{Th}
    \label{AX}
    Denote by $H_n=\int h_n$, then 
    \[
    \lim_{\hbar\to 0}H_n = H_n^c,\quad
    \fk{H_n}{H_m} = 0,\quad n,m\geq 0,
    \] 
    here the bracket is the quantum Hamiltonian structure defined in \eqref{AE}. 
\end{Th}

Finally, we consider the eigenvalue problem. As discussed before, we need to fix a representation of the quantum Hamiltonian structure. One of the motivations for constructing quantum dispersionless KdV hierarchy using a vertex algebra is the observation that any vertex algebra has a natural representation on its state space. Indeed, for any state $a\in V$, their modes $a_{(n)}$ are already endomorphisms on $V$. Since the state space is identified with the space $\mathcal A[\hbar]$ of differential polynomials, it follows that there is a family of actions on $\mathcal A[\hbar]$ for any given $f\in\mathcal A[\hbar]$, as described by the following proposition.
\begin{Prop}
    \label{AG}
Let $f\in\mathcal A[\hbar]$ be a differential polynomial. Then for $n\geq 0$, the $n$-th mode of the state $\varphi(f)$ induces an endomorphism on $\mathcal A[\hbar]$ which reads 
\begin{align*}
    \varphi^{-1}\comp\varphi(f)_{(n)}\comp \varphi = \sum_{m\geq 1}&\frac{\hbar^m}{m!}\sum_{\substack{r_1,\dots,r_m\geq 0\\ s_1,\dots,s_m\geq 0\\ 2m-1+\sum r_i+s_i\geq n}}\frac{(-1)^{\sum r_i}\prod (r_i+s_i+1)!}{(2m-1-n+\sum r_i+\sum s_i)!}\\
    &\cdot \qp_x^{2m-1-n+\sum r_i+\sum s_i}\kk{\frac{\qp^m f}{\qp v^{(r_1)}\dots\qp v^{(r_m)}}} \frac{\qp^m }{\qp v^{(s_1)}\dots\qp v^{(s_m)}}.
\end{align*}
\end{Prop}
\noindent
It follows that we can associate a family of operators for each $h_n$. To consider the eigenvalue problem it is essential to choose operators 
\[
\hat H_n =  \varphi^{-1}\comp\varphi(h_n)_{(n)}\comp \varphi\in\mathrm{End}(\mathcal A)[\hbar]
\] 
corresponding to the $n$-th mode of the state $\varphi(h_n)$ for $n\geq 0$, and the reason for this specific choice is explained in Sect.\,\ref{RAC}. We show that operators $\hat H_n$ mutually commute with each other, and their common eigenvectors are given by Schur polynomials, which are defined by
\[
S_\ql = \det{\kk{{\bf{e}}_{\ql_i-i+j}}_{1\leq i,j\leq \ell(\ql)}}\in\mathcal A[\sqrt{\hbar}],
\]
here $\ql = (\ql_1,\dots,\ql_{\ell(\ql)})$ with $\ql_1\geq\dots\geq\ql_{\ell(\ql)}>0$ is an integer partition of length $\ell(\ql)$, and ${\bf{e}}_n\in\mathcal A[\sqrt{\hbar}]$ is determined by the generating series
\[
\sum_{n\in\mathbb Z}{\bf{e}}_n z^n = \exp\kk{\sum_{k\geq 0} \frac{\hbar^{k/2} v^{(k)}}{(k+1)!}z^{k+1}}.
\]
\begin{Th}
    \label{AY}
    Schur polynomials $S_\ql$, where $\ql$ runs over all integer partitions, serve as a complete set of common eigenvectors of operators $\hat H_n$. More precisely, we have
    \[
    \hat H_n S_\ql = \frac{\hbar^{\frac{n+1}{2}}}{n+1}\sum_{k=0}^n\binom{n+1}{k}\kk{\binom{-\ell(\ql)}{k+1}+\sum_{i=1}^{\ell(\ql)}\binom{\ql_i-i}{k}}S_\ql,\quad \ql = (\ql_1,\dots,\ql_{\ell(\ql)}).
    \]
\end{Th}
\noindent
In studying these properties, we also arrive at the following explicit formula relating $h_n$ to Schur polynomials:
\begin{Prop}
    \label{AZ}
    We have 
    \[
    h_n = \frac{1}{n+1}\sum_{k=0}^n\frac{1}{\binom{n}{k}}S_{(n-k+1,\underbrace{1,\dots,1}_k)}.
    \]
\end{Prop}

\subsection*{Organization of the paper}
This paper is organized as follows. In Sect.\,\ref{AH}, after recalling some basic facts about vertex algebras,  we prove Theorem \ref{AF} and Proposition \ref{AG} in a slightly more general setting. In Sect.\,\ref{AI} we construct quantum Hamiltonians of the dispersionless KdV hierarchy and prove various properties. Our main tool for the proofs is based on the work \cite{liu2022action} studying actions of $W$-type operators on Schur functions.  Finally, in Sect.\,\ref{AJ}, we give some concluding remarks and compare our results to those obtained before \cite{buryak2016double,dubrovin2016symplectic}.

\subsection*{Acknowledgement} Part of the work was done during the JSPS International Research Fellowship of the author and was supported by JSPS KAKENHI Grant Number 23KF0114. He would like to thank Wenda Fang, Hiroshi Iritani, Si-Qi Liu, Youjin Zhang and Chunhui Zhou for very helpful discussions. He also thanks the anonymous referees for very helpful suggestions and comments to improve the
presentation of the paper.

\subsection*{Data Availability} Data sharing not applicable to this article as no datasets were generated or analyzed during
the current study.

\section{Heisenberg vertex algebra and quantum Hamiltonian structure}
\label{AH}
\subsection{Basics of vertex algebra}
\label{RAB}
Let us review some basic facts about vertex algebras, and we refer to \cite{frenkel2004vertex} for details. The paper \cite{de2006finite} also serves as a reference for the topic.

A vertex algebra is given by the data $\left(V,\vac,T,Y(-,z)\right)$, where $V$ is a vector space over $\mathbb C$ called the state space, $\vac\in V$ is a special vector called the vacuum vector, $T\in\mathrm{End}(V)$ is a linear operator called the infinitesimal translation operator and  $Y(-,z)\colon V\to \mathrm{End}(V)[[z^{\pm 1}]]$ is a linear morphism called the state-field correspondence. The state-field correspondence associates to each state $a\in V$ operators $a_{(n)}$ for all $n\in\mathbb Z$ given by  
\[
Y(a,z) = \sum_{n\in\mathbb Z}a_{(n)}z^{-n-1},\quad a_{(n)}\in\mathrm{End}(V).
\]
We call $a_{(n)}$ the $n$-th mode of $a$. The data $\left(V,\vac,T,Y(-,z)\right)$ of a vertex algebra are required to satisfy the following axioms:
\begin{itemize}
    \item $\vac_{(n)} = \qd_{n,-1}\mathrm{Id}_V$ and for any $a\in V$, the action of $a_{(n)}$ on the vacuum vector is given by 
     \[
     a_{(n)}\vac =\begin{cases}0,& \text{for } n\geq 0,\\
    \qd_{n,-1}a,& \text{for } n\leq -1.
    \end{cases}
     \]
     \item $T\vac = 0$ and for any $a\in V$, $\fk{T}{a_{(n)}} = -na_{(n-1)}$ and $Y(Ta,z) = \qp_z Y(a,z)$.
     \item For any $a,b\in V$, there exists a positive integer $N\in\mathbb Z_+$ such that 
      \[(z-w)^N[Y(a,z),Y(b,w)] = 0.\]
\end{itemize}
A convenient tool for actual computation of vertex algebras is the so-called $\ql$-bracket formalism (\cite{d1998structure,kac1998vertex}). For $a,b\in V$, we denote 
\[
[a_\ql b] = \sum_{n\geq 0}\frac{\ql^n}{n!}a_{(n)}b.
\]
It can be proved that the $\ql$-brackets are polynomials in $\ql$ with coefficients being states in $V$. They satisfy the following properties, where many useful identities for vertex algebras are encoded:
\begin{itemize}[itemindent=8em]
    \item[Sesquilinearity: ] $[Ta_\ql b] = -\ql[a_\ql b]$, and $ [a_\ql Tb]= (T+\ql)[a_\ql b]$.
    \item[Skewsymmetry: ] $[b_\ql a] = -[a_{-\ql-T}b]$.
    \item[Jacobi identity: ] $[a_\ql[b_\mu c]] = [b_\mu[a_\ql c]]+[[a_\ql b]_{\ql+\mu} c]$.
\end{itemize}

For $a,b\in V$, we denote by $:ab:$ the normal order product, which is defined to be the state $a_{(-1)}b$. For general vertex algebras, the normal order product is neither commutative nor associative. It satisfies the following  quasi-commutativity and quasi-associativity (\cite{de2006finite}):
\begin{align}
    \label{BC}
&:ab:-:ba: = \int_{-T}^0d\ql\,[a_\ql b],\\
\label{AM}
&::ab:c:-:a:bc:: = :\kk{\int_0^Td\ql\,a}[b_\ql c]:+:\kk{\int_0^Td\ql\,b}[a_\ql c]:.
\end{align}
To compute the right-hand sides of the above formulae, we must put $\ql$ on the left under the sign of integral, and act the operators obtained from the definite integral to the right. For example, we have
\[
    \int_{-T}^0d\ql\,[a_\ql b] = \sum_{n\geq 0}\kk{\int_{-T}^0 \frac{\ql^n}{n!}d\ql}\kk{a_{(n)}b} = \sum_{n\geq 0}\frac{(-1)^n T^{n+1}}{(n+1)!}\kk{a_{(n)}b}.
\]
When computing the $\ql$-bracket of normal order product, we can use the following useful formula:
\begin{equation}
    \label{BB}
[a_\ql :bc:] = :[a_\ql b]c:+:b[a_\ql c]:+\int_0^\ql [[a_\ql b]_\mu c] d\mu.
\end{equation}
Another standard fact about the normal order product is that
\begin{equation}
    \label{AL}
    Y(:ab:,z) = Y(a,z)_+Y(b,z)+Y(b,z)Y(a,z)_-,
\end{equation}
where we set
\[
    Y(a,z)_+ = \sum_{n<0}a_{(n)}z^{-n-1},\quad Y(a,z)_- = \sum_{n\geq 0}a_{(n)}z^{-n-1}.
\]

In this paper, we will only consider the Heisenberg vertex algebra which we define now. Recall that a Heisenberg Lie algebra of rank $N$ is a $\mathbb C$-vector space spanned by $b^\qa_{(n)}$ with $\qa = 1,\dots, N$ and $n\in\mathbb Z$, whose Lie bracket is specified by 
\begin{equation}
    \label{AK}
\left[b^\qa_{(n)},b^\qb_{(m)}\right] = \hbar\,\eta^{\qa\qb}\,n\,\qd_{m+n,0},\quad \qa,\qb = 1,\dots,N,\quad n,m\in\mathbb Z,
\end{equation} 
where $(\eta^{\qa\qb})$ is a constant non-degenerate symmetric matrix and $\hbar$ is a formal parameter.
\begin{Rem}
    Strictly speaking, the underlying space of the Heisenberg Lie algebra we define here is a free $\mathbb C[\hbar]$-module rather than a vector space. The only reason for the introduction of the parameter $\hbar$ is to make sense of classical limits by taking $\hbar\to 0$. In particular, all the definitions and axioms of vertex algebras make sense if the state space is a free $\mathbb C[\hbar]$-module and all endomorphisms are required to be $\hbar$-linear. 
\end{Rem}

We proceed to define the Heisenberg vertex algebra. First define the state space $V$ to be the free $\mathbb C[\hbar]$-module spanned by states of the form 
\[
b^{\qa_1}_{(-k_1-1)}\dots b^{\qa_n}_{(-k_n-1)}\vac,\quad n, k_i\geq 0,\quad \qa_i = 1,\dots, N.
\]
In particular, we have the vacuum vector $\vac\in V$, and we view $b^\qa_{(n)}$ as endomorphisms on $V$ whose actions are determined by the commutation relation \eqref{AK} and $b^\qa_{(n)}\vac = 0$ for $n\geq 0$. As an example, we have 
\[
b^\qa_{(n)}\kk{b^\qb_{(-k-1)}\vac} = 
\begin{cases}
    b^\qa_{(n)}b^\qb_{(-k-1)}\vac, & n<0\\
    \hbar\eta^{\qa\qb}n\qd_{n-k-1,0}\vac,& n\geq 0.
\end{cases}
\]
Note that we have $b^\qa_{(0)} = 0$. We define the infinitesimal translation operator $T$ similarly by specifying its action, which is determined by $T\vac = 0$ and the commutation relation $\fk{T}{b^\qa_{(n)}} = -nb^\qa_{(n-1)}$. Finally, we define
the state-field correspondence $Y(-,z)$ by specifying the modes of each state. For a state $b^{\qa_1}_{(-k_1-1)}\dots b^{\qa_n}_{(-k_n-1)}\vac$, its modes are defined recursively by induction on $n$ \cite{borcherds1986vertex}:
\begin{align*}
\vac_{(i)} &= \qd_{i,-1}\mathrm{Id}_V,\\
(b^\qa_{(i)}a)_{(j)} &= \sum_{\ell\geq 0}(-1)^\ell\binom{i}{\ell}\kk{b^\qa_{(i-\ell)}a_{(j+\ell)}-(-1)^i a_{(i+j-\ell)}b^\qa_{(\ell)}},\quad i,j\in\mathbb Z.
\end{align*}
From now on, we will denote by $b^\qa$ the state $b^\qa_{(-1)}\vac$, and it follows that
\[
Y(b^\qa,z) = \sum_{n\in\mathbb Z}b^\qa_{(n)}z^{-n-1}.
\]
\begin{Lem} The date $(V,\vac,T,Y(-,z))$ described above define a vertex algebra called the Heisenberg vertex algebra of rank $N$.
\end{Lem}
\begin{proof}
    We refer to Chapter 2 of \cite{frenkel2004vertex} for a proof.
\end{proof}

For the Heisenberg vertex algebra, the modes of states admit particularly nice formulae in terms of the normal ordering which we shall define now. We call $b^\qa_{(n)}$ an annihilator if $n\geq 0$ and a creator if $n<0$. Then, the normal ordering is defined for a sequence of modes of $b^\qa$, which is denoted by $:b^{\qa_1}_{(k_1)}\dots b^{\qa_n}_{(k_n)}:$, and its action is given by placing all annihilators on the right of all creators. It is important to note that the normal ordering is well-defined because creators commute with each other, so do annihilators. Now we proceed to give an explicit formula for the state-field correspondence of the Heisenberg vertex algebra.
\begin{Prop}[Remark 2.2.6 of \cite{frenkel2004vertex}]
    \label{AP}
    For $n\geq 0$, $k_1,\dots k_n\geq 0$ and $\qa_1,\dots,\qa_n = 1,\dots, N$, let us denote
    \[
    a = b^{\qa_1}_{(-k_1-1)}\dots b^{\qa_n}_{(-k_n-1)}\vac\in V,
    \]
    then we have
    \[
    Y\left(a,z\right) = (-1)^{\sum k_i}\sum_{m_1,\dots,m_n\in\mathbb Z}\binom{m_1}{k_1}\dots\binom{m_n}{k_n}:b^{\qa_1}_{(m_1-k_1)}\dots b^{\qa_n}_{(m_n-k_n)}:z^{-n-\sum m_i}.
    \]
\end{Prop}
\begin{proof}
    Let us prove by induction on $n$. The case $n=0$ trivially holds true. Assume that we have proved the proposition for some $n\geq 0$, and let us consider the state
    \[
        a = b^\qa_{(-k-1)}b^{\qa_1}_{(-k_1-1)}\dots b^{\qa_n}_{(-k_n-1)}\vac\in V
    \]  
    for some $k\geq 0$. First, it follows from $Y(T^kb^\qa,z) = \qp^k_z Y(b^\qa,z)$ that 
    \[
    b^\qa_{(-k-1)} = \frac{1}{k!}(T^k b^\qa)_{(-1)},
    \]
    which implies that
    \[
    a = \frac{1}{k!}:(T^k b^\qa)\kk{b^{\qa_1}_{(-k_1-1)}\dots b^{\qa_n}_{(-k_n-1)}\vac}:.
    \]
    The field corresponding to $a$ is then computed by using the identity \eqref{AL}, and we arrive at the following expression for $Y(a,z)$:
    \begin{align*}
     & \begin{split}
            (-1)^{k+\sum k_i}\sum_{\substack{m_1,\dots,m_n
            \in
            \mathbb Z\\ m<0}}&\binom{m}{k}\binom{m_1}{k_1}\dots\binom{m_n}{k_n}\\
            &\cdot b^\qa_{(m-k)}\kk{:b^{\qa_1}_{(m_1-k_1)}\dots b^{\qa_n}_{(m_n-k_n)}:}z^{-n-\sum m_i-m-1}
        \end{split}\\[1em]
        &\begin{split}
            +(-1)^{k+\sum k_i}\sum_{\substack{m_1,\dots,m_n\in\mathbb Z \\ m\geq 0}}&\binom{m}{k}\binom{m_1}{k_1}\dots\binom{m_n}{k_n}\\
            & \cdot \kk{:b^{\qa_1}_{(m_1-k_1)}\dots b^{\qa_n}_{(m_n-k_n)}:}b^\qa_{(m-k)}z^{-n-\sum m_i-m-1}.
        \end{split}
        \end{align*}
In the first summation, since we have $m-k<0$, it follows from the definition of the normal ordering that 
\[
    b^\qa_{(m-k)}\kk{:b^{\qa_1}_{(m_1-k_1)}\dots b^{\qa_n}_{(m_n-k_n)}:}\, =\,  :b^\qa_{(m-k)}b^{\qa_1}_{(m_1-k_1)}\dots b^{\qa_n}_{(m_n-k_n)}:,\quad m<0.
\]
In the second summation, terms for $0\leq m\leq k-1$ vanish due to the fact that the binomial coefficient $\binom{m}{k} = 0$, and for $m\geq k$ we have
\[
    \kk{:b^{\qa_1}_{(m_1-k_1)}\dots b^{\qa_n}_{(m_n-k_n)}:}b^\qa_{(m-k)} = :b^\qa_{(m-k)}b^{\qa_1}_{(m_1-k_1)}\dots b^{\qa_n}_{(m_n-k_n)}:.
\]
Therefore, we conclude that $Y(a,z)$ is indeed of the desired form and the proposition is proved.
\end{proof}
Using the above proposition as well as the definition of the  normal ordering, we arrive at the following description of modes.
\begin{Cor}
    \label{AO}
    Fix a state $a = b^{\qa_1}_{(-k_1-1)}\dots b^{\qa_n}_{(-k_n-1)}\vac$ for some $n\geq 0$ and $k_i\geq 0$, its modes are given by 
    \[
    a_{(m)} = \sum_{\substack{j_1\leq\dots\leq j_n\\\qb_1,\dots,\qb_n}} C^{j_1,\dots,j_n}_{\qb_1,\dots,\qb_n} b^{\qb_1}_{(j_1)}\dots b^{\qb_n}_{(j_n)}\in\mathrm{End}(V),\quad m\in\mathbb Z,
    \]
    here $C^{j_1,\dots,j_n}_{\qb_1,\dots,\qb_n}$ are some integers depending on both the state $a$ and $m$, where the indices $\qb_i$ run over $1,\dots, N$ and $j_i$ run over $\mathbb Z$.
\end{Cor}

Finally, let us recall the conformal degree of the Heisenberg vertex algebra. For a state 
\[
a = b^{\qa_1}_{(-k_1-1)}\dots b^{\qa_n}_{(-k_n-1)}\vac\in V,
\]
we call the integer
\[
\Qd_a = n+\sum_{i=1}^n k_i
\]
the conformal weight of $a$. Let us denote by $V^\Qd$ the subspace of $V$ consisting of homogeneous states of conformal weight $\Qd$, and we set $V^\Qd = 0$ if $\Qd<0$. An endomorphism is called of conformal degree $d$ if it maps $V^\Qd$ to $V^{\Qd+d}$ for any $\Qd\geq 0$. For example, it is easy to see that the mode $b^\qa_{(n)}$ is of conformal degree $-n$ for any $n\in\mathbb Z$.
\begin{Lem}
\label{RCON}
If $a\in V^\Qd$, then $a_{(n)}$ is of conformal degree $\Qd-n-1$ for any $n\in\mathbb Z$. In particular, if $a\in V^{\Qd_1}$ and $b\in V^{\Qd_2}$, then $:ab:\in V^{\Qd_1+\Qd_2}$.
\end{Lem}
\begin{proof}
This is an immediate consequence of Proposition\, \ref{AP}.
\end{proof}

\subsection{Quantum Hamiltonian structure}
In this section, we construct a deformation quantization of Hamiltonian structures of hydrodynamic types using the Heisenberg vertex algebra. To this end, we first recall the notion of a (classical) Hamiltonian structure, and we refer to \cite{liu2018lecture} for an introduction to the general theory of Hamiltonian structures.

Let $M$ be an $N$-dimensional smooth manifold, and let us denote by $\mathcal A$ the ring of differential polynomials of the infinite jet bundle $J^\infty M$. Locally, if we choose a coordinate system $(v^1,\dots,v^N)$ of $M$, then 
\[
\mathcal A = \mathbb C\left[v^{\qa,p}\colon \qa = 1,\dots,N;\,p\geq 0\right].
\]
From now on, we will fix such a local coordinate system. Define a derivation $\qp_x$ of $\mathcal A$ by 
\[
\qp_x = \sum_{p\geq 0} v^{\qa,p+1}\diff{}{v^{\qa,p}},
\]
here and henceforth we assume the summation over repeated lower and upper Greek indices. Then we denote by $\mathcal F = \mathcal A/\qp_x\mathcal A$ the space of local functionals, and for $f\in\mathcal A$, we denote by $\int f$ its class in $\mathcal F$. The ring $\mathcal A$ is graded by the differential degree defined by
\[
\deg_{\qp_x}v^{\qa,p} = p,
\]
and we denote by $\mathcal A_d$ the space of homogeneous differential polynomials of differential degree $d$.
\begin{Rem}
    In the literature, the ring of differential polynomials are defined by a certain completion of $\mathcal A$. For our purpose, this completion is not needed. 
\end{Rem}

A Hamiltonian structure is a Lie algebra structure on $\mathcal F$ whose Lie bracket is of the following special form:
\begin{equation*}
\{F,G\} = \int \sum_{k\geq 0}\vard{F}{v^\qa}P^{\qa\qb}_k\qp_x^k\vard{G}{v^\qb},\quad P^{\qa\qb}_k\in\mathcal A,\quad \forall F,G\in\mathcal F,
\end{equation*}
here the variational derivatives are defined by
\[
    \vard{F}{v^\qa} = \sum_{p\geq 0}(-\qp_x)^p\diff{f}{v^{\qa,p}},\quad F = \int f.
\]
Let us consider the following special Hamiltonian structure which is called the Hamiltonian structure of hydrodynamic type:
\begin{equation}
    \label{AN}
    \{F,G\} = \int\vard{F}{v^\qa}\eta^{\qa\qb}\qp_x\vard{G}{v^\qb},
\end{equation}
where $(\eta^{\qa\qb})$ is a non-degenerate symmetric constant matrix. It is proved \cite{getzler2002darboux} that the bracket \eqref{AN} serves as a universal local model for Hamiltonian structures with hydrodynamic leading terms. Such Hamiltonian structures are ubiquitous in the theory of integrable hierarchies. For example, all Dubrovin-Zhang hierarchies \cite{dubrovin2001normal} and all double ramification hierarchies \cite{buryak2015double} possess Hamiltonian structures of these kinds.

Let us fix a Hamiltonian structure of the form \eqref{AN} and let $V$ be the state space of the Heisenberg vertex algebra of rank $N$ as described in Sect.\,\ref{RAB}. Recall that we have states $b^1,\dots, b^N\in V$ whose modes satisfy the commutation relation
\[
    \left[b^\qa_{(n)},b^\qb_{(m)}\right] = \hbar\,\eta^{\qa\qb}\,n\,\qd_{m+n,0},\quad \qa,\qb = 1,\dots,N,\quad n,m\in\mathbb Z.
\]
The key observation is that we can identify the state space $V$ and the space of differential polynomials $\mathcal A$ via the $\hbar$-linear map given by
\begin{equation}
    \label{AU}
\varphi: \mathcal A[\hbar]\to V\colon v^{\qa_1,k_1}\dots v^{\qa_n,k_n}\mapsto k_1!\dots k_n! b^{\qa_1}_{(-k_1-1)}\dots  b^{\qa_n}_{(-k_n-1)}\vac.
\end{equation}
\begin{Lem}
    The map $\varphi$ defined in \eqref{AU} is an isomorphism between vector spaces. For any $f\in\mathcal A[\hbar]$, we have
    \begin{equation}
        \label{AS}
        \varphi\kk{\qp_x f} = T\varphi\kk{f},
    \end{equation}
    and for any $p\geq 0$, we have
    \begin{equation}
        \label{AT}
        \varphi\kk{v^{\qa,p}f} = p!\, b^\qa_{(-p-1)}\varphi(f),\quad \hbar\, \eta^{\qa\qb}\varphi\kk{\diff{f}{v^{\qb,p}}} = \frac{1}{(p+1)!}b^\qa_{(p+1)}\varphi(f).
    \end{equation}
\end{Lem}
\begin{proof}
    These facts are straightforward to verify, and we omit the details.
\end{proof}

It follows from \eqref{AS} that $\varphi$ induces an isomorphism between quotient spaces $\mathcal F[\hbar]$ and $V/TV$, and we will still denote this isomorphism by $\varphi$. It is well-known that $V/TV$ admits a natural Lie algebra structure given by $0$-th modes.
\begin{Lem}[Theorem 4.1.2 of \cite{frenkel2004vertex}]
Let $\pi\colon V\to V/TV$ be the projection map. Then the bracket
\begin{equation}
    \label{AV}
[A,B] = \pi\left(\pi^{-1}(B)_{(0)}\pi^{-1}(A)\right),\quad A,B\in V/TV
\end{equation} 
is well-defined, and defines a Lie algebra structure on $V/TV$.
\end{Lem}

Via the isomorphism $\varphi$, the above bracket on $V/TV$ can be identified to give a Lie bracket on $\mathcal F[\hbar]$. Our first goal is to compute this bracket in terms of differential polynomials explicitly and hence to prove that this bracket gives a deformation quantization of the Hamiltonian structure \eqref{AN}. For this purpose, let us compute the modes $\varphi(f)_{(n)}$ for $f\in\mathcal A$ and $n\geq 0$.

Recall that $V$ is a free $\mathbb C[\hbar]$-module, hence we can decompose 
\[
    V = \bigoplus_{n\geq 0}V_n,\quad V_n = \hbar^n V_0,
\]
where $V_0$ is the $\mathbb C$-vector space spanned by states of the form $ b^{\qa_1}_{(-k_1-1)}\dots b^{\qa_n}_{(-k_n-1)}\vac$. Then it is clear that for $a = b^{\qa_1}_{(-k_1-1)}\dots b^{\qa_n}_{(-k_n-1)}\vac\in V_0$, we have $(\hbar^m a)_{(n)} = \hbar^m\cdot a_{(n)}$, and furthermore the modes of $a$ can be decomposed as
\[
a_{(n)} = \sum_{i\geq 0}a_{(n)}^{[i]}, \quad a_{(n)}^{[i]}\colon V_0\to V_i,\quad n\in\mathbb Z.
\]
It follows from Corollary\, \ref{AO} that these modes are of the forms 
\[
    a_{(n)} = \sum_{j_1\leq\dots\leq j_n} C^{j_1,\dots,j_n}_{\qb_1,\dots,\qb_n} b^{\qb_1}_{(j_1)}\dots b^{\qb_n}_{(j_n)},
    \]
here the coefficients $C^{j_1,\dots,j_n}_{\qb_1,\dots,\qb_n}$ are constant numbers independent of $\hbar$. Then the $\hbar$-dependence of modes is given by the commutation relation \eqref{AK}, hence we arrive at
\begin{equation}
\label{RBA}
 a_{(n)}^{[i]} = \sum_{(j_1,\dots,j_n)\in J^i} C^{j_1,\dots,j_n}_{\qb_1,\dots,\qb_n} b^{\qb_1}_{(j_1)}\dots b^{\qb_n}_{(j_n)},
    \end{equation}
where we use $J^i$ to denote the index set
\[
\{
    (j_1,\dots,j_n)\in\mathbb Z^n\colon j_1\leq\dots\leq j_{n-i}\leq 0\leq j_{n-i+1}\leq\dots\leq j_n
\}.
\]

\begin{Lem}
    \label{AQ}
    Fix a state $a = b^{\qa_1}_{(-k_1-1)}\dots b^{\qa_n}_{(-k_n-1)}\vac\in V_0$ for some $n\geq 0$ and $k_i\geq 0$. Then we have
    \[
    a_{(-1)}^{[0]} =  b^{\qa_1}_{(-k_1-1)}\dots b^{\qa_n}_{(-k_n-1)}.
    \]
\end{Lem}
\begin{proof}
    It follows from Proposition\, \ref{AP} that 
    \[
        a_{(-1)} = (-1)^{\sum k_i}\sum_{m_1+\dots+m_n=-n}\binom{m_1}{k_1}\dots\binom{m_n}{k_n}:b^{\qa_1}_{(m_1-k_1)}\dots b^{\qa_n}_{(m_n-k_n)}:,
    \]
    and by further using \eqref{RBA}, we arrive at 
    \[
        a_{(-1)}^{[0]} = (-1)^{\sum k_i}\sum_{\substack{m_1+\dots+m_n=-n,\\ m_i\leq k_i}}\binom{m_1}{k_1}\dots\binom{m_n}{k_n}:b^{\qa_1}_{(m_1-k_1)}\dots b^{\qa_n}_{(m_n-k_n)}:.
    \]
    Since $b^{\qa_i}_{(0)} = 0$, it follows that the non-vanishing summand is with indices $m_i\leq -1$. Therefore, the only non-vanishing summand is given by the term with indices $m_1 = \dots=m_n=-1$, namely,
    \[
        a_{(-1)}^{[0]} = b^{\qa_1}_{(-k_1-1)}\dots b^{\qa_n}_{(-k_n-1)}.
    \]
    The lemma is proved.
\end{proof}

We have the following immediate consequence.
\begin{Cor}
    \label{AR}
    For $f,g\in \mathcal A$, we have 
    \[
    \varphi(f)_{(-1)}^{[0]}\varphi(g)=\varphi(fg).
    \]
\end{Cor}
\begin{proof}
This follows directly from the definition of $\varphi$. The corollary is proved.
\end{proof}

Now we are ready to prove the main result of this section.
\begin{Th}
    \label{BA}
Let $f,g\in\mathcal A$ be two differential polynomials. Then for $n\geq 0$, we have
\begin{align*}
    \varphi^{-1}\kk{\varphi(f)_{(n)}\varphi(g)} = \sum_{m\geq 1}&\frac{\hbar^m}{m!}\sum_{\substack{r_1,\dots,r_m\geq 0\\ s_1,\dots,s_m\geq 0\\ 2m-1+\sum r_i+s_i\geq n}}\frac{(-1)^{\sum r_i}\prod (r_i+s_i+1)!\eta^{\qa_i\qb_i}}{(2m-1-n+\sum r_i+\sum s_i)!}\\
    &\cdot\frac{\qp^m g}{v^{\qb_1,s_1}\dots v^{\qb_m,s_m}} \qp_x^{2m-1-n+\sum r_i+\sum s_i}\kk{\frac{\qp^m f}{\qp v^{\qa_1,r_1}\dots\qp v^{\qa_m,r_m}}}.
\end{align*}
\end{Th}
\begin{proof}
    By a similar argument as in the proof of Lemma\, \ref{AQ}, we see that $\varphi(f)_{(n)}^{[0]} = 0$ for $n\geq 0$. Therefore, we only need to consider $\varphi(f)_{(n)}^{[m]}$ for $m\geq 1$. Without loss of generality, we simply assume 
    \[\varphi(f) = b^{\qa_1}_{(-k_1-1)}\dots b^{\qa_\ell}_{(-k_\ell-1)}\vac\]
    for some $\ell\geq 0$ and $k_i\geq 0$.
    Using \eqref{RBA}, we may write 
    \[
        \varphi(f)_{(n)}^{[m]} = \sum_{(j_1,\dots,j_\ell)\in J^m} C^{j_1,\dots,j_\ell}_{\qb_1,\dots,\qb_\ell} b^{\qb_1}_{(j_1)}\dots b^{\qb_\ell}_{(j_\ell)},\quad C^{j_1,\dots,j_\ell}_{\qb_1,\dots,\qb_\ell}\in\mathbb Z.
    \]
    It is more convenient to 
    separate annihilators and creators in the above expression, and rewrite it into the form
    \[
        \varphi(f)_{(n)}^{[m]} = \sum_{k_1,\dots,k_m>0} A^{k_1,\dots,k_m}_{\qa_1,\dots,\qa_m}b^{\qa_1}_{(k_1)}\dots b^{\qa_m}_{(k_m)},\quad A^{k_1,\dots,k_m}_{\qa_1,\dots,\qa_m}\in\mathrm{End}(V),
    \]
    here $A^{k_1,\dots,k_m}_{\qa_1,\dots,\qa_m}$ are operators given by linear combinations of creators of the form
    \[
    b^{\qb_1}_{(-j_1-1)}\dots b^{\qb_{\ell-m}}_{(-j_{\ell-m}-1)},\quad j_1,\dots,j_{\ell-m}\geq 0.
    \]
    Since annihilators commute with each other, we may and do  fix each operator $A^{k_1,\dots,k_m}_{\qa_1,\dots,\qa_m}$ uniquely by requiring  that it is symmetric with respect to indices $(\qa_1,k_1),\dots,(\qa_m,k_m)$.
    Therefore, by using \eqref{AT}, we arrive at
    \begin{align*}
        \varphi(f)_{(n)}^{[m]}\varphi(g) = \hbar^m \sum_{k_1,\dots,k_m> 0} &k_1!\dots k_m!\, \eta^{\qa_1\qg_1}\dots\eta^{\qa_m\qg_m}\\
        &\cdot A^{k_1,\dots,k_m}_{\qa_1,\dots,\qa_m}\varphi\kk{\frac{\qp^m g}{\qp v^{\qg_1,k_1-1}\dots\qp v^{\qg_m,k_m-1}}}.
    \end{align*}
    In particular, it follows that 
    \begin{align*}
        \varphi(f)_{(n)}^{[m]}\varphi(v^{\qb_1,s_1-1}\dots v^{\qb_m,s_m-1}) =& 
            \hbar^m\sum_{k_1,\dots,k_m> 0} k_1!\dots k_m!\, \eta^{\qa_1\qg_1}\dots\eta^{\qa_m\qg_m}\\
        &\cdot A^{k_1,\dots,k_m}_{\qa_1,\dots,\qa_m} \sum_{\qs\in S_m}\prod_{i=1}^m \qd^{\qb_{\qs(i)}}_{\qg_i}\qd_{k_i,s_{\qs(i)}}\vac\\
        =\,&\hbar^m m!\,s_1!\dots s_m!\, \eta^{\qa_1\qb_1}\dots\eta^{\qa_m\qb_m}A^{s_1,\dots,s_m}_{\qa_1,\dots,\qa_m}\vac.
    \end{align*}
   Note that we have used the fact that $A^{k_1,\dots,k_m}_{\qa_1,\dots,\qa_m}$ is symmetric with respect to indices. Hence, by using Lemma\, \ref{AQ}, we have
    \begin{align*}
    \varphi(f)_{(n)}^{[m]}\varphi(g) = \frac{\hbar^m}{m!}\sum_{s_1,\dots,s_m\geq 0}&\kk{\varphi(f)_{(n)}^{[m]}\varphi(v^{\qb_1,s_1}\dots v^{\qb_m,s_m})}_{(-1)}^{[0]}\\
    &\cdot\varphi\kk{\frac{\qp^m g}{\qp v^{\qb_1,s_1}\dots\qp v^{\qb_m,s_m}}}.
    \end{align*}
    It follows from the skew-symmetry property of the $\ql$-bracket that for general $A,B\in V$, 
    \[
    A_{(n)}B = \sum_{k\geq 0}\frac{(-1)^{k+n+1}}{k!}T^k\kk{B_{(k+n)}A},
    \]
    hence to further simplify the above expression for $\varphi(f)_{(n)}^{[m]}\varphi(g)$ we only need to compute
    \[
        \varphi(v^{\qb_1,s_1}\dots v^{\qb_m,s_m})_{(k+n)}^{[m]}\varphi(f),\quad k\geq 0.
    \]
    By using the definition \eqref{AU} of $\varphi$ and Proposition\, \ref{AP}, we have
    \begin{align*}
        \varphi(&v^{\qb_1,s_1}\dots v^{\qb_m,s_m})_{(k+n)}^{[m]} = (-1)^{\sum s_i}s_1!\dots s_m!\\
        &\sum_{\substack{n_1,\dots,n_m\geq 0\\\sum n_i+\sum s_i+2m = k+n+1}}\binom{n_1+s_1+1}{s_1} \dots\binom{n_m+s_m+1}{s_m}b^{\qb_1}_{(n_1+1)}\dots b^{\qb_m}_{(n_m+1)}.
    \end{align*}
    Then after a straightforward computation, we arrive at
    \begin{align*}
        &\varphi(f)_{(n)}^{[m]}\varphi(v^{\qb_1,s_1}\dots v^{\qb_m,s_m})\\
        &=\sum_{\substack{r_1,\dots,r_m\geq 0\\\sum r_i+\sum s_i+2m \geq n+1}}\frac{(-1)^{\sum r_i}\prod_{i=1}^m (r_i+s_i+1)!\eta^{\qa_i\qb_i}}{\left(\sum r_i+\sum s_i+2m -n-1\right)!} \\
        &\cdot T^{\sum r_i+\sum s_i+2m -n-1}\kk{\varphi\left(\frac{\qp^m f}{\qp v^{\qa_1,r_1}\dots\qp v^{\qa_m,r_m}}\right)}.
    \end{align*} 
    Finally, we arrive at the desired formula by using the identity \eqref{AS} and Corollary\, \ref{AR}. The theorem is proved.
\end{proof}

Recall that we have a Lie bracket on $V/TV$ defined by the formula \eqref{AV}. By identifying $V/TV$ with $\mathcal F[\hbar]$ via $\varphi$, it follows from the above theorem that we have a Lie bracket on $\mathcal F[\hbar]$ defined by 
\begin{align*}
    [F,G] = \int \sum_{m\geq 1}\frac{\hbar^m}{m!}&\sum_{\substack{r_1,\dots,r_m\geq 0\\ s_1,\dots,s_m\geq 0}}\frac{(-1)^{\sum r_i}\prod (r_i+s_i+1)!\eta^{\qa_i\qb_i}}{(2m-1+\sum r_i+\sum s_i)!}\\
    &\cdot\frac{\qp^m f}{v^{\qb_1,s_1}\dots v^{\qb_m,s_m}} \qp_x^{2m-1+\sum r_i+\sum s_i}\kk{\frac{\qp^m g}{\qp v^{\qa_1,r_1}\dots\qp v^{\qa_m,r_m}}},
\end{align*}
where $f,g\in\mathcal A[\hbar]$ are arbitrary choices of densities of $F$ and $G$, respectively. This Lie bracket defines a deformation quantization of the Hamiltonian structure \eqref{AN} in the sense that 
\begin{equation}
    \label{BP}
\lim_{\hbar\to 0}\frac{1}{\hbar}[F,G] = \{F,G\},\quad F,G\in\mathcal F.
\end{equation}
In particular, we have proved Theorem\,\ref{AF} and Proposition\,\ref{AG} as special cases of Theorem\,\ref{BA} for $N=1$ and $\eta = 1$.

As an application of the above formula, let us prove the following interesting fact which will be used later.
\begin{Prop}
    \label{BO}
    Let $a\in V$ be a state with $a_{(0)} = 0$. Then there exist polynomials $B_\qa(\hbar), C(\hbar)\in\mathbb C[\hbar]$ such that 
    \[
    a = B_\qa(\hbar)b^\qa+C(\hbar)\vac+TV.
    \]
\end{Prop}
\begin{proof}
    Let us denote by $f = \varphi^{-1}(a)\in\mathcal A[\hbar]$, and we decompose 
    \[
    f = \sum_{n\geq 0}\hbar^n f_n, \quad f_n\in\mathcal A.
    \]
    The assumption $a_{(0)} = 0$ implies that for any $g\in \mathcal A$, $\varphi(f)_{(0)}\varphi(g) = 0$.
    Then by taking the $\hbar$-coefficient of $\varphi(f)_{(0)}\varphi(g)$, we arrive at 
    \begin{equation}
        \label{BQ}
    \sum_{r,s\geq 0}\eta^{\qa\qb}(-1)^r\diff{g}{v^{\qb,s}}\qp_x^{r+s+1}\diff{f_0}{v^{\qa,r}} = 0.
    \end{equation}
    In particular, we have 
    \[
        \{G,F_0\} = \int \vard{G}{v^\qa}\eta^{\qa\qb}\qp_x\vard{F_0}{v^\qb}=0,\quad F_0=\int f_0,\quad G = \int g.
    \]
    Then it follows from Lemma 2.1.7 of \cite{liu2011jacobi} that there exist constants $A$ and $B^\qa$ such that 
    \[
    \eta^{\qa\qb}\vard{f_0}{v^\qb} = A v^\qa+B^\qa,
    \]
    which implies that there exists another constant $C$ such that 
    \[
    f_0 = \frac 12 \eta_{\qa\qb}A v^\qa v^\qb+\eta_{\qa\qb}B^\qa v^\qb+C+\qp_x\mathcal A,
    \]
    here we denote by $\eta_{\qa\qb}$ the inverse of $\eta^{\qa\qb}$. Due to \eqref{BQ}, we see that $A=0$, therefore we conclude that 
    \[
        f_0 = \eta_{\qa\qb}B^\qa v^\qb+C+\qp_x\mathcal A,
    \]
    which implies that actually $\varphi(f_0)_{(0)} = 0$. Hence, we replace $f$ by $(f-f_0)/\hbar$, and the same argument shows that $f_1$ is of the same form as $f_0$. In this way, we prove that actually each $f_n$ has the same form of $f_0$, namely, there exist polynomials $B_\qa(\hbar), C(\hbar)\in\mathbb C[\hbar]$ such that 
    \[
    f = B_\qa(\hbar)v^\qa+C(\hbar)+\qp_x\mathcal A[\hbar].
    \]
    The proposition is proved. 
\end{proof}

\section{Weyl quantization of dispersionless KdV hierarchy}
\label{AI}
In this section, let us consider the quantization of the dispersionless KdV hierarchy. To state the problem in a precise way, recall that the classical dispersionless KdV hierarchy is a family of compatible PDEs
\[
    \diff{v}{t_n} = \frac{v^n}{n!}\diff{v}{x},\quad n\geq 0
\] 
for the unknown function $v$. This system admits a Hamiltonian structure given by
\[
\{F,G\} = \int \vard{F}{v}\qp_x\vard{G}{v},\quad F,G\in\mathcal F,
\]
and each flow $\diff{}{t_n}$ is the Hamiltonian vector filed with Hamiltonian $H_{n+1}^c$, where
\[
H_n^c = \int \frac{v^{n+1}}{(n+1)!},\quad n\geq 0.
\]
It follows that these Hamiltonians are in involution:
\[
\{H_n^c,H_m^c\} = 0,\quad n,m\geq 0.
\]
To consider the quantization problem, we replace the Hamiltonian structure by the deformed bracket. Then the classical Hamiltonians do not commute anymore. For example, we have
\[
[H_2^c,H_3^c] = \int \frac{\hbar}{4}v^4v_x+\frac{\hbar^2}{12}(3vv_xv_{xx}+v^2v_{xxx})+\frac{\hbar^3}{720}v^{(5)}\neq 0.
\]
The quantization problem is then to find quantum Hamiltonians $H_n\in\mathcal F[\hbar]$ such that 
\[
[H_n,H_m] = 0,\quad \lim_{\hbar\to 0}H_n = H_n^c.
\]

\subsection{Construction of quantum Hamiltonians}
\label{AD}
From now on, we will consider the Heisenberg vertex algebra of rank 1 with $\eta=1$. Our idea to construct the quantum Hamiltonians is as follows. The ring $\mathcal A$ of differential polynomials  is a commutative algebra, and we can view the normal order product on $V$ as a deformation of this algebra structure, in the sense that 
\[
\lim_{\hbar \to 0} :\varphi(f)\varphi(g): = \varphi(fg),\quad f,g\in\mathcal A,
\]
which follows from Corollary\, \ref{AR}. The normal order product is neither commutative nor associative, and we can use the idea of Weyl quantization in quantum mechanics to quantize the Hamiltonians $H_n^c$. Roughly speaking, in Weyl quantization, a quantum Hamiltonian is obtained from a classical Hamiltonian by averaging all the possible ways to compose all the quantum operators appearing in Hamiltonian. For example, let $(x,p)$ be a pair of conjugated variables, then we have the following Weyl quantization:
\[
(xp)\string^ = \frac{1}{2}(\hat x\hat p+\hat p\hat x),\quad (xp^2)\string^ = \frac{1}{3}(\hat x\hat p\hat p+\hat p\hat x\hat p+\hat p\hat p\hat x).
\]
In our case, we quantize the Hamiltonians $H_n^c$ by averaging all the possible ways to associate the state $b=\varphi(v)$ using the normal order product. For example, we should define 
\[
h_2 = \frac{1}{12}\varphi^{-1}(:b:bb::+::bb:b:),\quad H_2 = \int h_2.
\]
We call this procedure a non-associative Weyl quantization.

To give a precise definition of the non-associative Weyl quantization of Hamiltonians, let us define a set $BT(n)$ of $(n-1)!$ elements for all $n\geq 1$. Each element in $BT(n)$ is a full binary tree whose nodes, except for the root node, are labeled either by $L$ (for left) or $R$ (for right) and the root node is labeled by an integer $i$ with $1\leq i\leq (n-1)!$. In addition, all $n$ leaves of a tree in $BT(n)$ are labeled by $1,\dots,n$.

We proceed to define $BT(n)$ recursively. First, define
$BT(1)$ to be the set $\{1\}$, where we interpret the element $1\in BT(1)$ as a single root node labeled by 1. Define $BT(2)$ to be the set of the unique full binary tree with two leaves labeled by $L,1$ and $R,2$ and the root labeled by 1.

\begin{figure}[h]
    \centering
    \begin{tikzpicture}
        \node {1}
           child {node{L,1}}
           child {node{R,2}};
      \end{tikzpicture}
      \caption{The unique tree in $BT(2)$.}
\end{figure}
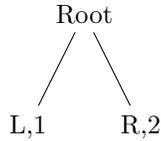

To define the set $BT(n+1)$ for $n\geq 2$, we attach a leaf node to trees in $BT(n)$ and change the labels accordingly. Take a tree $T\in BT(n)$ whose root node is labeled by $k$ and fix   a choice of its leaf with label $i \in\{1,\dots,n\}$, let us define a new tree $T_i$ as follows. Firstly, we make the chosen leaf with label $i$ an internal node by attaching to it 2 leaves, and hence obtain a full binary tree with $n+1$ leaves. Then the label `$i$' of this node is removed. For leaves originally labeled by $1,\dots,i-1$ their labels remain unchanged, and for leaves originally labeled by $i+1,\dots,n$ their labels are increased by one. Finally, the newly added leaves are labeled as $L, i$ and $R, i+1$, and the label of the root is changed from $k$ to $(k-1)n+i$. We then define
\[
    BT(n+1):=\coprod_{T\in BT(n)}\{T_i\mid i=1,\dots,n\}.
\]
As examples, the elements in $BT(3)$ are presented in Fig.\,\ref{fig2} and the elements in $BT(4)$ are presented in Fig.\,\ref{fig4} and Fig.\,\ref{fig5}.

\begin{figure}[h]
    \centering
    \begin{tikzpicture}
        \node {1}
           child {node {L}
             child {node {L,1}}
             child {node {R,2}}
           }
           child {node {R,3}
           };
      \end{tikzpicture} 
      \begin{tikzpicture}
        \node {2}
           child {node {L,1}
           }
           child {node {R}
           child {node {L,2}}
           child {node {R,3}}
           };
      \end{tikzpicture}
      \caption{Trees in $BT(3)$. If we denote by $T$ the unique tree in $BT(2)$, then the first tree is $T_1$ and the second one is $T_2$.}
  \label{fig2}
\end{figure}
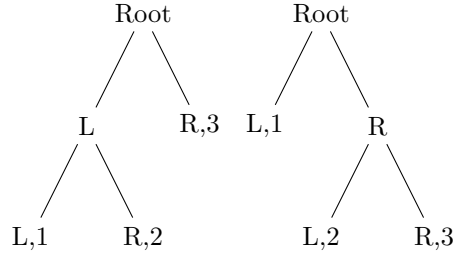
\begin{Rem}
In $BT(n)$, there are distinct elements with the same binary tree structure. For example, we see that in $BT(4)$, there are only 5 distinct trees. However, there are 6 distinct elements in $BT(4)$ as labeled trees. This is also the reason why $BT(n)$ consists of $(n-1)!$ elements rather than just $C_n$ elements, where $C_n$ is the $n$-th Catalan number counting the number of full binary trees with $n$ leaves.
\end{Rem}
Let us proceed to associate a state $A_T$ to each labeled tree  $T\in BT(n)$. For $1\in BT(1)$, we simply set $A_1 = b$. For $n\geq 2$ and $T\in BT(n)$, we first place the state $b$ on every leaf, and then recursively place the state at all other nodes obtained by taking the normal order product of states associated with its left child and right child. More precisely, given a node which is not a leaf, assume that we have put $f\in V$ at its left child and $g\in V$ at its right child, then this node is associated with the state $:fg:\in V$. We denote by $A_T\in V$ the state associated with the root of $T$. For example, the state associated to the unique tree in $BT(2)$ is $:bb:$, and for $T_1,T_2\in BT(3)$ as presented in Fig.\,\ref{fig2}, we have
\[
A_{T_1} = ::bb:b:,\quad A_{T_2} = :b:bb::.
\]  
\begin{figure}[h] 
    \centering  
    \begin{tikzpicture}
        \node {$::bb:b:$}
           child {node {$:bb:$}
             child {node {$b$}}
             child {node {$b$}}
           }
           child {node {$b$}
           };
      \end{tikzpicture} 
      \begin{tikzpicture}
        \node {$:b :bb::$}
           child {node {$b$}
           }
           child {node {$:bb:$}
           child {node {$b$}}
           child {node {$b$}}
           };
      \end{tikzpicture}
      \caption{Evaluation of each node for trees in $BT(3)$.}
\end{figure}
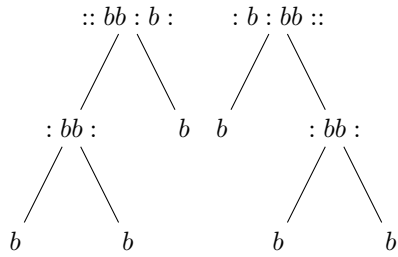

Let us denote by $Z_n$ the state obtained by averaging all the states $A_T$ for $T\in BT(n)$, namely, we define
\[
Z_n = \frac{1}{(n-1)!}\sum_{T\in BT(n)}A_T,\quad n\geq 1.
\]
The following lemma is a straightforward consequence of our definition.
\begin{Lem}
\label{RLD}
    For $n\geq 2$, we have
    \[
Z_n = \frac{1}{n-1}\sum_{i=1}^{n-1}:Z_iZ_{n-i}:.
    \]
\end{Lem}
\begin{proof}
We observe that each tree in $BT(n)$ can be decomposed into two subtrees by removing the root node. Then this lemma is verified by a straightforward computation.
\end{proof}

By using the above notation, we define differential polynomials
\[
h_n = \frac{1}{(n+1)!}\varphi^{-1}(Z_{n+1}),\quad n\geq 0.
\]
It follows from Corollary\,\ref{AR} that 
\[
\lim_{\hbar\to 0}\varphi^{-1}(Z_n) = v^n,
\]
therefore the Hamiltonians defined by 
$H_n = \int h_n$ is a deformation of classical Hamiltonians $H_n^c$ in the sense that 
\[
\lim_{\hbar\to 0}H_n = H_n^c.
\]
By replacing $Z_i$ by $i!\varphi(h_{i-1})$ in Lemma\, \ref{RLD}, it follows that we can equivalently define $h_n$ recursively by
\begin{equation}
    \label{BD}
h_0 = \varphi^{-1}(b),\quad h_n = \frac{1}{n}\sum_{k=0}^{n-1}\frac{1}{\binom{n+1}{k+1}}\varphi^{-1}\left(:\varphi(h_k)\varphi(h_{n-k-1}):\right),\quad n\geq 1.
\end{equation}

\begin{Ex}
    It is easy to verify that $h_0 = v$ and $h_1 = \frac{v^2}{2}$. Let us compute $h_2$. First it follows from the definition of normal order product that
    \[
    :b:bb:: = \hes{-1}\hes{-1}\hes{-1}\vac,
    \]
    and by using the quasi-associativity \eqref{AM}, we have 
    \[
    ::bb:b: = :b:bb::+:\kk{\int_0^Td\ql\,b}[b_\ql b]: = :b:bb::+\hbar T^2 b,
    \] 
    here we use the fact that 
    \[
        [b_\ql b] = \sum_{n\geq 0}\frac{\ql^n}{n!}\hes{n}\hes{-1}\vac = \hbar \ql\vac.
    \]
    Therefore, it follows that 
    \[
    Z_3 = \frac{1}{2}\left(::bb:b:+:b:bb::\right) = b^3+\frac{\hbar}{2} T^2b,\quad b^3 = :b:bb::,
    \]
    and hence
    \[
    h_2 = \frac{1}{6}v^3+\frac{\hbar}{12}v_{xx}.
    \]
\end{Ex}
\begin{Ex}
    Let us compute $h_3$. We begin by enumerate elements in $BT(4)$. Denote by $R$ and $S$ the first and the second tree presented in Fig.\,\ref{fig2}, then by definition we arrive at 
    \[
        BT(4) = \{R_1,R_2,R_3,S_1,S_2,S_3\}.
        \]
    These trees are presented in Fig.\,\ref{fig4} and Fig.\,\ref{fig5}.
    
    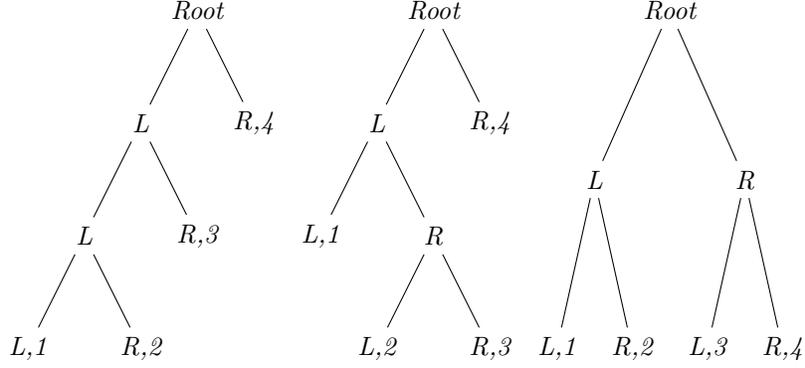
\begin{figure}[h]
        \centering
        \begin{tikzpicture}
            \node {1}
               child {node {L}
                 child {node {L}
                    child {node{L,1}}
                    child {node{R,2}}
                 }
                 child {node {R,3}}
               }
               child {node {R,4}
               };
          \end{tikzpicture} 
          \begin{tikzpicture}
            \node {2}
               child {node {L}
                 child {node {L,1}}
                    child { node{R}
                    child {node{L,2}}
                    child{node{R,3}}
                    }
                 }
               child {node {R,4}
               };
          \end{tikzpicture}
          \begin{tikzpicture}
            [   
                level distance=22.5mm,
                level 1/.style={sibling distance=20mm},
                level 2/.style={sibling distance=10mm}
            ]
            \node {3}
               child {node {L}
                    child{node{L,1}}
                    child{node{R,2}}
                }
               child {node {R}
                    child{node{L,3}}
                    child{node{R,4}}
               };
          \end{tikzpicture}
          \caption{Trees $R_1$, $R_2$ and $R_3$.}
          \label{fig4}
    \end{figure}
    \begin{figure}[h]
        \centering
        \begin{tikzpicture}
            [   
                level distance=22.5mm,
                level 1/.style={sibling distance=20mm},
                level 2/.style={sibling distance=10mm}
            ]
            \node {4}
               child {node {L}
                    child{node{L,1}}
                    child{node{R,2}}
                }
               child {node {R}
                    child{node{L,3}}
                    child{node{R,4}}
               };
        \end{tikzpicture}
        \begin{tikzpicture}
            \node {5}
            child{node{L,1}}
            child{node{R}
                child{node{L}
                    child{node{L,2}}
                    child{node{R,3}}
                }
                child{node{R,4}}
            };
        \end{tikzpicture} 
        \begin{tikzpicture}
            \node {6}
            child{node{L,1}}
            child{node{R}
            child{node{L,2}}
            child{node{R}
                child{node{L,3}}
                child{node{R,4}}
            }
            };
        \end{tikzpicture}
          \caption{Trees $S_1$, $S_2$ and $S_3$.}
          \label{fig5}
    \end{figure}
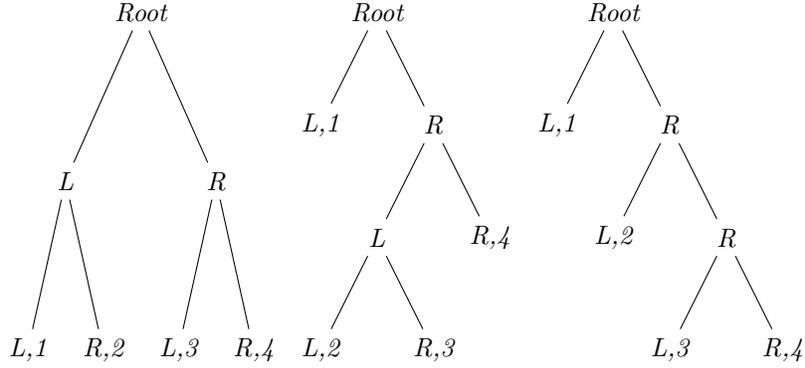
    
    Let us compute the contribution from each graph. First we note that
    \[A_{R_3} = A_{S_1} = :b^2b^2:,\quad b^2=\,:bb:.\]
    By using formula \eqref{AM}, it follows that 
    \[
        :b^2b^2:-b^4 = 2\left(\int_0^T d\ql\,b\right)[b_\ql b^2] =2\hbar :T^2bb:,\quad b^4=\,:b:b:bb:::,
        \]
    here we use \eqref{BB} to obtain 
    \[
        [b_\ql b^2] = :[b_\ql b]b:+:b [b_\ql b]:+\int_0^\ql[[b_\ql b]_\mu b]d\mu = 2\hbar\ql b.
        \]
    From the above computation we arrive at
    \[A_{R_3} = A_{S_1} = b^4+2\hbar :T^2bb:.\] 
    Next it is easy to obtain that 
    \[
        A_{S_2} = :b:b^2 b:: = :b\left(b^3+\hbar T^2 b\right):\,= b^4+\hbar:b T^2b:.
        \]
    Note that it follows from the quasi-commutativity \eqref{BC} that 
    \[
    :bT^2b:\, = \,:T^2bb:.
    \]
    To compute $A_{R_2}$, we may use the following quasi-commutativity formula:
    \[
        b^4-:b^3b: = \int_{-T}^0d\ql[b_\ql b^3] = \int_{-T}^0d\ql (3\hbar\ql b^2) = -3\hbar(:TbTb:+:bT^2b:).
        \] 
    Hence, we arrive at 
    \[
        A_{R_2} = :b^3b: = b^4+3\hbar(:TbTb:+:bT^2b:).
        \]
    Finally, we have $A_{S_3} = b^4$ and 
    \[
        A_{R_1} = ::b^2b:b: = :(b^3+\hbar T^2b)b: = b^4+3\hbar:TbTb:+4\hbar:bT^2b:.
        \]
    By averaging the contribution from each tree, we have
    \[
        Z_4 = b^4+\hbar :TbTb:+2\hbar:bT^2b:,
        \]
    which gives
    \[
    h_3 = \frac{v^4}{24}+\frac{\hbar}{24}\kk{v_x^2+2vv_{xx}}.
    \]
\end{Ex}
\begin{Ex}
    By a similar but more complicated computation, we compute $h_4$ by averaging contributions from 24 trees and $h_5$ by averaging 120 trees, and the result is 
    \begin{align*}
        h_4 &= \frac{v^5}{120}+\frac{\hbar}{24}\kk{vv_x^2+v^2v_{xx}}+\frac{\hbar^2}{360}v^{(4)},\\
        h_5 &=\frac{v^6}{720}+\frac{\hbar}{72}\kk{\frac32 v^2v_x^2+v^3v_{xx}}+\frac{\hbar^2}{360}\kk{vv^{(4)}+2 v_xv_{xxx}+\frac32 v_{xx}^2}.
    \end{align*}
\end{Ex}

\begin{Lem}
\label{RHOM}
    We have the following homogeneous properties for $h_n$:
    \begin{align*}
        \kk{\sum_{k\geq 0}(k+1)v^{(k)}\diff{}{v^{(k)}}}h_n &= (n+1)h_n,\\
        \kk{-2\hbar\diff{}{\hbar}+\sum_{k\geq 0}k v^{(k)}\diff{}{v^{(k)}}}h_n &= 0,\quad n\geq 0.
    \end{align*}
\end{Lem}
\begin{proof}
    It follows from the definition and Lemma\,\ref{RCON} that $Z_n$ is of conformal weight $n$. The conformal degree of $V$ induces a gradation on $\mathcal A[\hbar]$ via $\varphi$, with $v^{(k)}$ of degree $k+1$, hence we prove the first identity.
 
 For the second identity, we are to consider the differential degree of $h_n$. We define the differential degree of $\mathcal A[\hbar]$ to be 
 \[
\deg_{\qp_x} v^{(n)} = n,\quad \deg_{\qp_x}\hbar = -2.
 \] 
 This gradation induces a gradation on $V$, where $\deg_{\qp_x}\hes{n} = -n-1$. For $f\in\mathcal A[\hbar]$ of differential degree $d$, it is easy to see, using Proposition\,\ref{AP} that 
 \[
 \deg_{\qp_x}\varphi(f)_{(n)} = d-n-1.
 \] 
 Then we can prove by induction, using recursion \eqref{BD}, that each $h_n$ is of differential degree $0$. The lemma is proved.
\end{proof}

\subsection{Another expression for quantum Hamiltonians}
As we have already seen, it is not easy to compute each $h_n$ using the definition. In what follows, we give another expression of $h_n$, which simplifies computations and more importantly, enables us to prove various properties of quantum Hamiltonians. From now on, for simplicity, we will set $\hbar = 1$, and the dependence of $\hbar$ can be recovered from homogeneous conditions.

We start by defining states $p^{(n)}\in V$ for $n\geq 0$:
\[
p^{(0)} = b,\quad p^{(n+1)} = \frac{1}{n+2}\kk{T p^{(n)}+:b p^{(n)}:},\quad n\geq 0.
\]
Denote by $g_n = \varphi^{-1}(p^{(n)})\in\mathcal{A}$, then by using relations \eqref{AS} and \eqref{AT}, it is easy to see that $g_n$ satisfies the recursion relation
\[
g_0 = v,\quad g_{n+1} = \frac{1}{n+2}\kk{\qp_x g_n+vg_n},\quad n\geq 0.
\]

This recursion relation makes it easy to compute each $g_n$, and in particular, it follows that $g_n$ are given by Schur polynomials. Recall that for an integer partition $\ql = (\ql_1,\dots,\ql_{\ell(\ql)})$, we associate it with a differential polynomial $S_\ql\in\mathcal A$, called the Schur polynomial associated with the partition $\ql$, by
\[
S_\ql = \det{\kk{{\bf{e}}_{\ql_i-i+j}}_{1\leq i,j\leq \ell(\ql)}},
\]
where ${\bf{e}}_n\in\mathcal A$ is given by
\[
\sum_{n\in\mathbb Z}{\bf{e}}_n z^n = \exp\kk{\sum_{k\geq 0} \frac{v^{(k)}}{(k+1)!}z^{k+1}}.
\] 

\begin{Lem}
    \label{BG}
    We have $g_n = S_{(n+1)}$ for $n\geq 0$.
\end{Lem}
\begin{proof}
    Let us denote by 
    \[
    G(z) = \sum_{n\geq 0}g_nz^n
    \]
    the generating function of $g_n$, then it follows from the recursion relation of $g_n$ that this function is the unique formal series solution of the equation
    \[
     z\diff{G}{z}= z(\qp_x G+v G)-G+v,\quad G(0) = v.
    \]
    Therefore, it suffices to show that the formal power series
    \[
    \frac{W(z)}{z}-\frac{1}{z}
    \]
    solves the above equation, with
    \[
    W(z) = \exp\kk{\sum_{k\geq 0} \frac{v^{(k)}}{(k+1)!}z^{k+1}}.
    \]
    Indeed, it follows that 
    \begin{align*}
    &\qp_x\kk{\frac{W}{z}-\frac{1}{z}} = \frac{1}{z}\sum_{k\geq 0}v^{(k+1)}\diff{W}{v^{(k)}} = \frac Wz \kk{\sum_{k\geq 0} \frac{v^{(k+1)}}{(k+1)!}z^{k+1}},\\
    &z\diff{}{z}\kk{\frac{W}{z}-\frac{1}{z}} = \diff{W}{z}-\frac{W}{z}+\frac{1}{z} = vW+\qp_x W-\frac{W}{z}+\frac{1}{z}.
    \end{align*}
    Then the lemma is proved by a straightforward verification. 
\end{proof}
\begin{Rem}
    We will give another proof of the above lemma in the next subsection based on the definition of Hall-Littlewood polynomials.
\end{Rem}
The differential polynomials $g_n$ are closely related to $h_n$. In fact, we have the following key observation which makes it possible to prove various properties of $h_n$.
\begin{Th}
    \label{BE}
    For $n\geq 0$, we have 
    \[
h_n = \sum_{k=0}^n\frac{(-1)^{n-k}}{(n-k+1)!}\qp_x^{n-k}g_k.
    \]
\end{Th}
We will prove this theorem in the next subsection.
\subsection{Properties of quantum Hamiltonians}
Let us prove various properties of quantum Hamiltonian densities $h_n$ of quantum dispersionless KdV hierarchy. Our main technique is based on the vertex operator representation of W-type operators and their relations to Hall-Littlewood polynomials, which is studied in \cite{jing1991vertex}. These operators have a deep relationship with and find wide application in the theory of (classical) integrable hierarchies such as KP theory and its reduction (see, e.g., \cite{liu2022action} and references therein), and we are now to apply these techniques to study the quantum integrable hierarchy. 

We start by making necessary preparations. Denote by 
\[
J(z) = \sum_{n\geq 0}\frac{v^{(n)}}{n!}z^n+\sum_{n\geq 0}(n+1)!\diff{}{v^{(n)}}z^{-n-2}\in\mathrm{End}(\mathcal A)[[z^{\pm 1}]],
\]
then we see that $J(z)$ is precisely 
\[
J(z) = \varphi^{-1}\comp Y(b,z)\comp\varphi.
\]
Let us then define
\[
P^{(k)}(z) = \frac{1}{(k+1)!}:(\qp_z+J(z))^k J(z): = \sum_{m\in\mathbb Z}P^{(k)}_m z^{-k-m-1},
\]
here the normal order product means to put the annihilators, i.e. the operators $\diff{}{v^{(n)}}$ on the right. Then it follows that  
\[
    P^{(k)}(z) = \varphi^{-1}\comp Y\left(\varphi(g_k),z\right)\comp\varphi,
\]
in particular, it follows from the relation $\varphi(g_k)_{(-1)}\vac = \varphi(g_k)$ that 
\[
g_k = P^{(k)}_{-k-1}(1).
\]
Denote by 
\[
\phi(z) = \sum_{n\geq 1}\frac{v^{(n-1)}}{n!}z^n-\sum_{n\geq 1}(n-1)!\diff{}{v^{(n-1)}}z^{-n},
\]
then we have (\cite{liu2022action})
\begin{equation}
\label{BF}
    P^{(k)}(z) = \frac{1}{(k+1)!}:e^{-\phi(z)}\qp_z^{k+1}e^{\phi(z)}:.
\end{equation}
We define $B_n,B^*_n \in\mathrm{End}(\mathcal A)$ by 
\[
B(z) = :e^{\phi(z)}: = \sum_{n\in\mathbb Z} B_nz^n,\quad B^*(z) = :e^{-\phi(z)}: = \sum_{n\in\mathbb Z} B_n^*z^{-n},
\] 
these operators satisfy the relation (\cite{jing1991vertex})
\begin{align}
    \label{BN}
B_mB_n+B_{n-1}B_{m+1} &= 0,\\
B_m^*B_n^*+B_{n+1}^*B_{m-1}^* &= 0,\\
\label{BH}
B_mB_n^*+B_{n-1}^*B_{m-1} &= \qd_{m,n}.
\end{align} 
We can represent each operator $P^{(k)}_m$ via $B_n$ and $B_n^*$ by using \eqref{BF}.
\begin{Lem}[Theorem 3.3 in \cite{liu2022action}]
    \label{BR}
    The following relation holds true for any $k\geq 0$ and $m\in\mathbb Z$:
    \[P^{(k)}_m = -\sum_{n\in\mathbb Z}\binom{n}{k} B^*_{m+n}B_n.\]
\end{Lem}
Operators $B_n$ can be used to define the Schur polynomials (\cite{macdonald1998symmetric}). Given an integer partition $\ql = (\ql_1,\dots,\ql_{\ell(\ql)})$, it follows that 
\[
S_\ql = B_{\ql_1}\comp\dots\comp B_{\ql_{\ell(\ql)}}(1),
\]
furthermore, the following relation also holds true (identity (72) of \cite{liu2022action}):
\[
B^*_{-n}(1) = (-1)^n S_{(\underbrace{1,\dots,1}_n)}.
\]
It is much more convenient to work with Hall-Littlewood polynomials, namely for any $(\ql_1,\dots,\ql_\ell)\in\mathbb Z^\ell$, we define
\[
    S_\ql = B_{\ql_1}\comp\dots\comp B_{\ql_{\ell}}(1)\in\mathcal A.
\] 
It follows that for $\ql_1\geq \ql_2\dots\geq \ql_\ell>0$, the corresponding Hall-Littlewood polynomial coincides with the Schur polynomial.

Using the above ingredients, for $n\geq 0$, we arrive at 
\begin{align}
    \notag
P^{(\ell)}_{-n-1}(1)&=-\sum_{a\in\mathbb Z}\binom{a}{\ell} B^*_{a-n-1}B_a(1)\\
\notag
&=\sum_{a=0}^n\binom{a}{\ell} B_{a+1}B_{a-n}^*(1)\\
\notag
&=\sum_{a=0}^n(-1)^{a-n}\binom{a}{\ell} B_{a+1}S_{(\underbrace{1,\dots,1}_{n-a})}\\
\label{BI}
&=\sum_{a=0}^n(-1)^{a-n}\binom{a}{\ell} S_{(a+1,\underbrace{1,\dots,1}_{n-a})},
\end{align}
here for the second equality we use the relation \eqref{BH} and the fact that $B^*_{k}(1) = 0$ for $k>0$.
In particular, by setting $\ell = n$, it follows that 
\[
    g_n = P^{(n)}_{-n-1}(1) = S_{(n+1)},
\]
which gives another proof of Lemma\,\ref{BG}.

Now we are ready to study properties of quantum Hamiltonians. Our strategy is as follows: first we define $\tilde h_n\in\mathcal A$ by 
\begin{equation}
    \label{BJ}
    \tilde h_n = \sum_{k=0}^n\frac{(-1)^{n-k}}{(n-k+1)!}\qp_x^{n-k}g_k,
\end{equation}
then we prove Proposition\,\ref{AW}, Theorem\,\ref{AX}, Theorem\,\ref{AY} and Proposition\,\ref{AZ} with $h_n$ replaced by $\tilde h_n$. Finally, we prove that $\tilde h_n$ actually satisfies the recursion relation \eqref{BD}, hence $\tilde h_n=h_n$, and this proves Theorem\,\ref{BE}.

\subsubsection{\textbf{Proposition\,\ref{AZ} for $\tilde h_n$}}
We start by expressing $\tilde h_n$ via Schur polynomials. It follows from the definition \eqref{BJ} that 
\begin{equation}
    \label{BL}
\varphi^{-1}\comp Y(\varphi(\tilde h_n),z)\comp \varphi = \sum_{k=0}^n\frac{(-1)^{n-k}}{(n-k+1)!}\qp_z^{n-k}P^{(k)}(z),
\end{equation}
then we arrive at the relation
\[
\tilde h_n = \sum_{k=0}^n\frac{(-1)^{k}}{k+1}P^{(n-k)}_{-n-1}(1).
\]
Applying the identity \eqref{BI}, it turns out that 
\begin{align}
    \notag
    \tilde h_n &= \sum_{k=0}^n\frac{(-1)^{k}}{k+1}\sum_{a=n-k}^n (-1)^{a-n}\binom{a}{n-k}S_{(a+1,\underbrace{1,\dots,1}_{n-a})} \\
    \label{BK}
    &= \frac{1}{n+1}\sum_{k=0}^n\frac{1}{\binom{n}{k}}S_{(n-k+1,\underbrace{1,\dots,1}_k)},
\end{align}
here we use the following binomial identity proved in Corollary 2.2 of \cite{sury2004identities}:
\[
\sum_{k=0}^n\frac{(-1)^k}{m+k+1}\binom{n}{k} = \frac{1}{(n+m+1)\binom{n+m}{m}},\quad n,m\in\mathbb Z_{\geq 0}.
\]
In this way, we prove Proposition\,\ref{AZ} for $\tilde h_n$.

\subsubsection{\textbf{Proposition\,\ref{AW} for $\tilde h_n$}}
Let us proceed to prove that differential polynomials $\tilde h_n$ satisfy the recursion relations \eqref{BU}. Denote by
\[
\mathcal D = 2\,\varphi^{-1}\comp \varphi(\tilde h_2)_{(1)}\comp\varphi\in \mathrm{End}(\mathcal A),
\]
then it follows from Proposition\,\ref{AG} that 
\[
\mathcal D = \sum_{s,t\geq 0}\frac{(s+t+1)!}{s!t!}v^{(s)}v^{(t)}\diff{}{v^{(s+t)}}+ \frac{(s+1)!(t+1)!}{(s+t+2)!}v^{(s+t+2)}\frac{\qp^2}{\qp v^{(s)}\qp v^{(t)}}.
\]
Equivalently, using \eqref{BL}, this operator can also be written as 
\[
\mathcal D = P^{(1)}_{-1}+2P^{(2)}_{-1},
\]
and we need to prove that 
\[
\mathcal D(\tilde h_n) = (n+1)(n+2)\tilde h_{n+1},\quad n\geq 0.
\]
The actions of $P^{(k)}_m$ on Hall-Littlewood polynomials are computed in \cite{liu2022action}. Specifically, fix $(\ql_1,\dots,\ql_\ell)\in\mathbb Z^\ell$, then for any $k
\geq 0$ and $m\in\mathbb Z$, we have
\begin{align}
    \notag
    P^{(k)}_m(S_\ql) = &\sum_{i=1}^\ell\binom{\ql_i-m-i}{k}S_{(\ql_1,\dots,\ql_i-m,\dots,\ql_\ell)}+\qd_{m,0}\binom{-\ell}{k+1} S_\ql\\
    \label{BM}
    &+\qd_{m<0}\sum_{i=1}^{-m}(-1)^{m-i}\binom{i-\ell-1}{k}S_{(\ql_1,\dots,\ql_\ell,i,\underbrace{1,\dots,1}_{-m-i})},
\end{align}
where we set $\qd_{m<0} = 1$ for $m<0$ and $\qd_{m<0} = 0$ for $m\geq 0$. 
Let us use this to compute $\mathcal D(\tilde h_n)$. Due to \eqref{BK}, we first consider the action 
\[
P^{(1)}_{-1}S_{(n-k+1,\underbrace{1,\dots,1}_k)}
\]
for $0\leq k\leq n$. Using the relation \eqref{BN}, we have $B_1B_2 = 0$, hence 
\[
S_{(n-k+1,1,\dots,1,2,1,\dots,1)} = 0.
\]
Then it follows that
\[
    P^{(1)}_{-1}S_{(n-k+1,\underbrace{1,\dots,1}_k)} = (n-k+1)S_{(n-k+2,\underbrace{1,\dots,1}_k)}-(k+1)S_{(n-k+1,\underbrace{1,\dots,1}_{k+1})},
\]
and similarly we have 
\[
    2P^{(2)}_{-1}S_{(n-k+1,\underbrace{1,\dots,1}_k)} = (n-k+1)(n-k)S_{(n-k+2,\underbrace{1,\dots,1}_k)}+(k+1)(k+2)S_{(n-k+1,\underbrace{1,\dots,1}_{k+1})}.
\]
We arrive at the following computation:
\begin{align*}
    \mathcal D(\tilde h_n)=&\,\frac{1}{n+1}\sum_{k=0}^n\frac{1}{\binom{n}{k}}\mathcal D(S_{(n-k+1,\underbrace{1,\dots,1}_k)})\\
    =&\,\frac{1}{n+1}\sum_{k=0}^n\frac{(n-k+1)^2}{\binom{n}{k}}S_{(n-k+2,\underbrace{1,\dots,1}_k)}\\
    &+\frac{1}{n+1}\sum_{k=0}^n\frac{(k+1)^2}{\binom{n}{k}}S_{(n-k+1,\underbrace{1,\dots,1}_{k+1})}\\
    =&\,(n+1)\sum_{k=0}^{n+1}\frac{1}{\binom{n+1}{k}}S_{(n-k+2,\underbrace{1,\dots,1}_k)}\\
    =&\,(n+1)(n+2)\tilde h_{n+1}.
\end{align*}
This proves Proposition\,\ref{AW} for $\tilde h_n$.
\subsubsection{\textbf{Theorem\,\ref{AY} for $\tilde h_n$}}
\label{RAC}
Before we prove Theorem\,\ref{AY} for $\tilde h_n$, let us explain why it is necessary to take $\varphi(h_n)_{(n)}\in\mathrm{End}(V)$ for studying the eigenvalue problem. We have proved in Lemma\,\ref{RHOM} that $\varphi(h_n)\in V$ is of conformal weight $n+1$, and it follows from Lemma\,\ref{RCON} that the operator $\varphi(h_n)_{(m)}$ is of conformal degree $n-m$. It follows that $\varphi(h_n)_{(m)}$ admits non-zero eigenvectors unless $m = n$.

Let us proceed to study the eigenvalue problems for $\tilde h_n$. By using \eqref{BM}, we see that all the Schur polynomials $S_\ql$ serve as common eigenvectors of $P^{(k)}_0$. It follows from \eqref{BL} that 
\[
\varphi^{-1}\comp\varphi(\tilde h_n)_{(n)}\comp \varphi = \frac{1}{n+1}\sum_{k=0}^n\binom{n+1}{k}P^{(k)}_0.
\]
Then we prove Proposition\,\ref{AY} for $\tilde h_n$ by applying \eqref{BM}. Furthermore, we recover the $\hbar$-dependence by noticing that $\varphi(\tilde h_n)_{(n)}$ is of differential degree $-n-1$ and $S_\ql$ is of differential degree zero. 

\subsubsection{\textbf{Theorem\,\ref{AX} for $\tilde h_n$}}
Next let us prove that the functionals $\tilde H_n=\int\tilde h_n$ are in involution with respect to the deformed quantum Hamiltonian structure, which is equivalent to prove 
\[
\varphi(\tilde h_n)_{(0)}\varphi(\tilde h_m)\in TV,\quad m,n\geq 0.
\]
\begin{Prop}
    For $m,n\geq 0$, we have 
    \[
    \left[\varphi(\tilde h_n)_{(0)},\varphi(\tilde h_m)_{(0)}\right] = 0.
    \]
\end{Prop}
\begin{proof}
By using the relation \eqref{BL}, it suffices to show that     
\[
    \left[P^{(n)}_{-n},P^{(m)}_{-m}\right] = 0,\quad m,n\geq 0.
    \]
It follows from Lemma\,\ref{BR} that 
\begin{align*}
    \left[P^{(n)}_{-n},P^{(m)}_{-m}\right] =\,&\sum_{a,b\in\mathbb Z}\binom{a}{m}\binom{b}{n}\left[B^*_{a-n}B_a,B^*_{b-m}B_b\right]\\
    =\,&\sum_{a,b\in\mathbb Z}\binom{a}{m}\binom{b}{n}B^*_{a-n}B_aB^*_{b-m}B_b-B^*_{b-m}B_bB^*_{a-n}B_a\\
    =\,&\sum_{a\in\mathbb Z} \binom{a}{n}\binom{a+m}{m}B^*_{a-n}B_{a+m}-\sum_{a\in\mathbb Z} \binom{a}{n}\binom{a-n}{m}B^*_{a-n-m}B_{a}\\
    =\,&0,
\end{align*}
here for the third equality, we use the relations \eqref{BN} --\eqref{BH}. The proposition is proved. 
\end{proof}
By applying the well-known identity
\begin{equation}
    \label{BS}
    \left[A_{(n)},B_{(m)}\right] = \sum_{i\geq 0}\binom{n}{i}(A_{(i)}B)_{(m+n-i)},\quad A,B\in V, \quad n,m\in\mathbb Z, 
\end{equation}
we conclude that 
\begin{equation}
    \label{BT}
    \kk{\varphi(\tilde h_n)_{(0)}\varphi(\tilde h_m)}_{(0)} = 0,
\end{equation}
hence it follows from Proposition\,\ref{BO} that there exist constants $B_{n,m}$  and $C_{n,m}$ such that 
\[
    \varphi(\tilde h_n)_{(0)}\varphi(\tilde h_m)=B_{n,m}b+C_{n,m}\vac+TV.
\]
We notice that each $\varphi(\tilde h_n)$ is of conformal weight $n+1$, hence $\varphi(\tilde h_n)_{(0)}\varphi(\tilde h_m)$ is of conformal weight $n+m+1$, which implies that $C_{n,m} = 0$ and for $(m,n)\neq (0,0)$ the constants $B_{n,m}$ also vanish, namely
\[
    \varphi(\tilde h_n)_{(0)}\varphi(\tilde h_m)\in TV,\quad (m,n)\neq (0,0).
\]
But for $m = n = 0$, we know $\varphi(\tilde h_0) = b$ and it follows that $\varphi(\tilde h_0)_{(0)} = 0$, from which we conclude that the above identity also holds true in this case. This proves Theorem\,\ref{AX} for $\tilde h_n$.

\subsubsection{\textbf{Proof of Theorem\,\eqref{BE}}}
Finally, we are to prove  $\tilde h_n = h_n$. Since we already have $\tilde h_0 = h_0 = v$, it suffices to prove the validity of the recursion relation 
\[
    \varphi(\tilde h_{n+1}) = \frac{1}{n+1}\sum_{k=0}^{n}\frac{1}{\binom{n+2}{k+1}}\varphi(\tilde h_k)_{(-1)}\varphi(\tilde h_{n-k}),\quad n\geq 0.
\]
Let us verify the above relation by induction on $n$. The case $n=0$ can be proved directly and we assume that we have proved this for $n = 0,1,\dots,N-1$ for some $N\geq 1$. Now let us consider the case $n = N$. Firstly, it follows from the recursion 
\[
    \mathcal D(\tilde h_N)  = (N+1)(N+2)\tilde h_{N+1}
\]
that 
\[
    \varphi(\tilde h_{N+1}) = \frac{2}{(N+1)(N+2)}\varphi(\tilde h_2)_{(1)}\varphi(\tilde h_{N}).
\]
Therefore, by using the induction hypothesis, we arrive at 
\begin{align*}
    \varphi(\tilde h_{N+1}) &= \frac{2}{(N+1)(N+2)}\varphi(\tilde h_2)_{(1)}\varphi(\tilde h_{N})\\
    &=\frac{2}{N(N+1)(N+2)}\sum_{k=0}^{N-1}\frac{1}{\binom{N+1}{k+1}}\varphi(\tilde h_2)_{(1)}\varphi(\tilde h_k)_{(-1)}\varphi(\tilde h_{N-1-k}).
\end{align*}
By combining \eqref{BS} and \eqref{BT}, we arrive at 
\[
    \fk{\varphi(\tilde h_2)_{(1)}}{\varphi(\tilde h_k)_{(-1)}} = \kk{\varphi(\tilde h_2)_{(1)}\varphi(\tilde h_k)}_{(-1)}.
\]
Finally, we can compute straightforwardly as 
\begin{align*}
    \varphi(\tilde h_{N+1}) =\,& \frac{2}{N(N+1)(N+2)}\sum_{k=0}^{N-1}\frac{1}{\binom{N+1}{k+1}}\varphi(\tilde h_2)_{(1)}\varphi(\tilde h_k)_{(-1)}\varphi(\tilde h_{N-1-k})\\
    =\,&\frac{2}{N(N+1)(N+2)}\sum_{k=0}^{N-1}\frac{1}{\binom{N+1}{k+1}}\varphi(\tilde h_k)_{(-1)}\varphi(\tilde h_2)_{(1)}\varphi(\tilde h_{N-1-k})\\
    &+\frac{2}{N(N+1)(N+2)}\sum_{k=0}^{N-1}\frac{1}{\binom{N+1}{k+1}}\kk{\varphi(\tilde h_2)_{(1)}\varphi(\tilde h_k)}_{(-1)}\varphi(\tilde h_{N-1-k})\\
    =&\,\frac{1}{N(N+1)(N+2)}\sum_{k=0}^{N-1}\frac{(N-k)(N-k+1)}{\binom{N+1}{k+1}}\varphi(\tilde h_k)_{(-1)}\varphi(\tilde h_{N-k})\\
    &+\frac{1}{N(N+1)(N+2)}\sum_{k=0}^{N-1}\frac{(k+1)(k+2)}{\binom{N+1}{k+1}}\varphi(\tilde h_{k+1})_{(-1)}\varphi(\tilde h_{N-1-k})\\
    =&\frac{1}{N+1}\sum_{k=0}^{N}\frac{1}{\binom{N+2}{k+1}}\varphi(\tilde h_k)_{(-1)}\varphi(\tilde h_{N-k}),
\end{align*}
and we conclude that $\tilde h_n = h_n$.

\section{Conclusion}\label{AJ}

In this paper, we study the relationship between vertex algebras and quantum integrable hierarchies. This idea is quite natural, indeed, Barakat, De Sole, Kac and their collaborators developed a theory to describe Hamiltonian integrable hierarchies via Poisson vertex algebras (see, e.g., \cite{kac2017introduction}), and these algebras, in some cases, can be viewed as a classical limit of vertex algebras. In this sense, the vertex algebras should provide suitable tools to study quantum integrable hierarchies, and we investigate this possibility in this paper via the example of quantum dispersionless KdV hierarchy. An explicit deformation quantization of Hamiltonian structures of hydrodynamic type is constructed, and the corresponding quantum Hamiltonians are constructed for this hierarchy. All these constructions are based on the structure theory of the Heisenberg vertex algebra.

Let us compare the results of the present paper to those obtained in the papers \cite{buryak2016double,dubrovin2016symplectic}. First, we briefly recall their construction of quantum Hamiltonian structures and quantum Hamiltonians for quantum dispersionless KdV hierarchy. The classical Hamiltonian structure \eqref{AC} can be equivalently represented by 
\begin{equation}
    \label{AB}
\{v(x),v(y)\} = \qd'(x-y).
\end{equation}
This Poisson bracket admits a canonical quantization by expanding the field $v(x)$ into Fourier modes 
\[
v(x) = \sum_{k\in\mathbb Z}v_k e^{ikx},
\]
and rewriting the Poisson bracket \eqref{AB} into the form 
\[
\{v_m,v_n\} = i\,m\,\qd_{m+n,0}.
\]
Then, the Fourier modes $v_n$ are quantized to creators $\hat q_n$ for $n\geq 1$ and to annihilators $\hat p_{-n}$ for $n\leq -1$, and the mode $v_0$ is quantized to 0 for simplicity. These quantum operators satisfy the commutating relation
\begin{equation}
    \label{AAA}
[\hat p_n,\hat q_m] = \hbar\,n\,\qd_{m,n},\quad [\hat p_n,\hat p_m]=[\hat q_n,\hat q_m]=0,\quad n,m\geq 1,
\end{equation}
where $\hbar$ is a formal parameter. An explicit formula of the above quantum Hamiltonian structure in terms of differential polynomials is given in \cite{buryak2016double}, and it reads 
\begin{align*}
    [F,G] = \int \sum_{m\geq 1}\frac{\hbar^m}{m!}&\sum_{\substack{r_1,\dots,r_m\geq 0\\ s_1,\dots,s_m\geq 0}}\frac{\qp^m f}{v^{(s_1)}\dots v^{(s_m)}}(-1)^{\sum r_i}\\
    &\cdot \sum_{j=1}^{2m-1+\sum r_i+\sum s_i}C_j^{r_1+s_1+1,\dots,r_m+s_m+1}\qp_x^j\kk{\frac{\qp^m g}{\qp v^{(r_1)}\dots\qp v^{(r_m)}}},
\end{align*}
where $F,G\in\mathcal F$ with $f,g\in\mathcal A$ a choice of their densities, and $C_j^{a_1,\dots,a_m}$ are certain constants. In particular, it is proved that
\[
C^{a_1,\dots,a_m}_{m-1+\sum a_i} = \frac{a_1!\dots a_m!}{(m-1+\sum a_i)!}.
\]
We compare this bracket to the one we obtained in Theorem\,\ref{AF}, and conclude that the above bracket can be viewed as a further deformation of the quantum Hamiltonian structure constructed in the present paper. 

The quantum Hamiltonians constructed and studied in \cite{buryak2016double,dubrovin2016symplectic} are also different from those presented in this paper. Denote by $H_n^{BR} = \int h_n^{BR}$ the quantum Hamiltonians constructed by Buryak and Rossi, then it is given by the following generating series \cite{dubrovin2016symplectic,rossi2008gromov}:
\[
 h^{BR}(z) =1+\sum_{n\geq 0}  h_n^{BR}z^{n+1}= \exp\left(zs(\sqrt{\hbar}z\qp_x)v(x)\right),
\]
where the function
\[
s(t)=\frac{e^{t/2}-e^{-t/2}}{t} = \sum_{n\geq 0}\frac{t^{2n}}{4^n(2n+1)!}.
\]
Here we slightly modify the original choices of densities presented in \cite{buryak2016double,dubrovin2016symplectic} for a clearer comparison. Indeed, 
we compare them with the quantum Hamiltonian densities constructed in the present paper, and they read  
\begin{align*}
    & h_0 = h_0^{BR} = v;\\
    & h_1 = h_1^{BR} = \frac12 v^2;\\
&h_2 = \frac16 v^3+\frac{\hbar}{12}v_{xx},\quad h_2^{BR} =  \frac16 v^3+\frac{\hbar}{24}v_{xx};\\
&h_3 = \frac{v^4}{24}+\frac{\hbar}{24}\kk{v_x^2+2vv_{xx}},\quad h_3^{BR} = \frac{v^4}{24}+\frac{\hbar}{24}vv_{xx}.
\end{align*}
The choices of quantum densities are different, however, we  note that as local functionals, $H_n = H_n^{BR}$ in the above examples. Actually, this can be checked for more examples, and we conjecture this to hold true for any $n$. This implies a close relation between two different ways of quantizations.

Another difference is that, in the construction of \cite{buryak2016double,dubrovin2016symplectic}, what matters is the quantum Hamiltonians, instead of quantum Hamiltonian densities. Indeed, the eigenvalue problem considered in \cite{dubrovin2016symplectic} only depends on the deformed Hamiltonians as local functionals, and does not depend on the choice of densities. However, in our construction, the densities $h_n$ are important since we are considering operators 
\[
\varphi^{-1}\comp\varphi(h_n)_{(n)}\comp\varphi,
\]
and we cannot add arbitrary terms in $\qp_x\mathcal A$ to $h_n$. 

Despite these differences, our results also share similarity to the known results. Besides the conjectural relation $H_n = H_n^{BR}$, the common eigenvectors of $H_n^{BR}$ computed in \cite{dubrovin2016symplectic} are also given by Schur polynomials, and further properties of quantum KdV hierarchy studied in \cite{van2024quantum,ittersum2024quantum,ruzza2021spectral} seem to be applicable to the present setting. Therefore, it makes sense to conjecture that our formulation is equivalent to the formulation presented in \cite{buryak2016double,dubrovin2016symplectic} in some sense. For example, can we actually recover the results in \cite{dubrovin2016symplectic} using results of the present paper? Is the bracket constructed in \cite{buryak2016double} related to the bracket presented in this paper by a certain transformation? These problems are still under investigation.

\end{document}